%% file: main.tex
\documentclass[runningheads]{llncs}

\usepackage[T1]{fontenc}
\usepackage{hyperref}
\usepackage{amsmath}
\usepackage{amssymb}
\usepackage{bm}
\usepackage{stmaryrd}
\usepackage{mathtools}
\usepackage{mathpartir}
\usepackage{graphicx}
\usepackage{braket}
\usepackage{algorithm}
\usepackage{algpseudocode}
\usepackage{booktabs}
\usepackage{multirow}
\usepackage{subcaption}
\usepackage{tikz}
\usetikzlibrary{shapes.geometric, arrows}
\usepackage{listings}
\usepackage{proof}
\usepackage{xcolor}

\usepackage{orcidlink}             %

\AtEndEnvironment{definition}{\qed\enddefinition}
\AtEndEnvironment{example}{\qed\endexample}
\AtEndEnvironment{remark}{\qed\endremark}

\title{Automated Catamorphism Synthesis for Solving Constrained Horn Clauses over Algebraic Data Types}
\titlerunning{Catamorphism Synthesis for Program Invariants over ADT}

\author{Hiroyuki Katsura\inst{1}\orcidlink{0000-0003-3420-4207} \and Naoki Kobayashi\inst{1}\orcidlink{0000-0002-0537-0604} \and
Ken Sakayori\inst{1}\orcidlink{0000-0003-3238-9279} \and
Ryosuke Sato\inst{2}\orcidlink{0000-0001-8679-2747}}
\institute{The University of Tokyo, Tokyo, Japan \\
\email{\{h.katsura,koba,sakayori\}@is.s.u-tokyo.ac.jp}
\and
Tokyo University of Agriculture and Technology, Tokyo, Japan \\
\email{rsato@acm.org}
}
\authorrunning{Hiroyuki Katsura \and Naoki Kobayashi \and Ken Sakayori, and Ryosuke Sato}

\input{preamble}

\newif\iffull\fulltrue
\newif\iflong\longtrue
\newif\ifdraft\draftfalse
\ifdraft
  \newcommand{\katsura}[1]{\textcolor{blue}{[#1 -katsura]}}
  \newcommand{\nk}[1]{\textcolor{red}{[#1 -nk]}}
  \newcommand*{\ks}[1]{\textcolor{teal}{\scriptsize [#1 -ks]}}
  \newcommand{\todo}[1]{\textcolor{blue}{[TODO: #1]}}
\else
  \newcommand{\katsura}[1]{}
  \newcommand{\nk}[1]{}
  \newcommand*{\ks}[1]{}
  \newcommand{\todo}[1]{}
\fi
\usepackage[utf8]{inputenc}
\begin{document}
\maketitle

\begin{abstract}
  We propose a novel approach to satisfiability checking of Constrained Horn Clauses (CHCs) %
  over Algebraic Data Types (ADTs).
  CHC-based automated verification has gained considerable attention in recent years, leading to the development of various %
  CHC solvers.
  However, existing solvers for CHCs over ADTs are not fully satisfactory, due to their limited ability to find and express
   models
   involving
   inductively defined functions/predicates (e.g., those about the sum of list elements).
To address this limitation, we consider \emph{catamorphisms} (generalized fold functions), and present a framework for automatically discovering appropriate catamorphisms on demand and using them to express a model of given CHCs.
We have implemented a new CHC solver called Catalia based on the proposed method.
Our experimental results for the CHC-COMP 2024 benchmark
show that Catalia outperforms state-of-the-art solvers in solving satisfiable CHCs over ADTs. Catalia was also used as a core part of the tool called ChocoCatalia, which won
the ADT-LIA category of CHC-COMP 2025.

\end{abstract}

\input{intro2.tex}
\input{pre.tex}

\input{overview3.tex}

\input{catalia2.tex}

\input{synthesis3.tex}

\input{disc}
\input{eval.tex}

\input{related.tex}

\input{conc.tex}

\subsubsection*{Acknowledgments}
    We would like to thank anonymous reviewers for useful comments.
    This work was supported by JSPS KAKENHI Grant Numbers JP23KJ0546 and JP20H05703.

\bibliographystyle{splncs04}
\bibliography{katsura}
\appendix
\iflong
\input{soundness2}
\input{app_exp}
\fi

\end{document}

%% file: preamble.tex
\newcommand \hoice{{\sc HoIce}}
\newcommand \eldarica{{\sc Eldarica}}
\newcommand \spacer{{\sc Spacer}}

\newcommand \catalia {{\sc Catalia}}
\newcommand \ringen{{\sc RInGen}}

\newcommand \Int {\mathbb{Z}}
\newcommand \Def {\stackrel{\Delta}{=}}
\newcommand \BNFDef {::=}
\DeclareMathOperator{\dom}{dom}

\DeclareMathOperator{\fv}{FV}

\newcommand\sem[1]{\llbracket #1 \rrbracket}

\newcommand\seq[1]{\widetilde{#1}}
\newcommand{\slambda}{%
  \mathop{%
    \rlap{$\lambda$}%
    \mkern2mu
    \raisebox{.275ex}{$\lambda$}%
  }%
}

\newcommand\sort {\xi}
\newcommand\dsort {\delta}
\newcommand \TADTZ {T_{\text{ADT+Z}}}
\newcommand \TZ {T_{\text{Z}}}

\newcommand \PREDSYM {P}

\newcommand \CONS {C}
\newcommand \sym[1] {\mathit{#1}}

\newcommand \universe[1] {#1}
\newcommand \herbrand {\mathcal{H}}
\newcommand \dterm {t}

\newcommand \cons {\sym{cons}}
\newcommand \nil {\sym{nil}}
\newcommand \ilist {\sym{ilist}}
\newcommand \nat {\sym{nat}}

\newcommand \combine {\mathcal{F}}
\newcommand \NAPPROX {N}
\newcommand \cata {\mathcal{C}}
\newcommand \cataT {\overline{\mathcal{C}}}
\newcommand {\banana}[1] {\cata^{#1}}
\newcommand {\bananaT}[2] {\cataT^{#1}_{#2}}
\newcommand \cataeo {\cata_{eo}}
\newcommand \defaultCATA{\cata_0}

\newcommand \cataA {\alpha_{\cata}}
\newcommand \cataAalt[1] {\alpha_{#1}}

\newcommand \CENV {\Xi}
\newcommand \cataAT {{\alpha}_{\cata, \CENV}}

\newcommand \cataAFB {\alpha_{\cata,\CENV}}

\newcommand \CODP {\AVAR{P}_\delta}
\newcommand \CODPo {P_\delta}

\newcommand \cataTMP {\mathcal{T}}
\newcommand \linearTMP {\mathcal{L}}
\newcommand \templates {\mathfrak{T}}

\newcommand \encoderalt[2] {\langle#1\rangle_{#2}}

\newcommand \AVAR[1] {\overline{#1}}

\newcommand \cache {\Theta}

\newcommand \form{\varphi}

\newcommand \true {\textbf{tt}}

\newcommand \false {\textbf{ff}}
\newcommand \FALSE {\false}
\newcommand\HFalse{\FALSE}
\newcommand \predicate {\textbf{p}}

\newcommand \arith {\textbf{a}}
\newcommand \operator {\mathbin{\textbf{op}}}
\newcommand \term {M}

\newcommand \constraint {\theta}

\newcommand \clause {\mathbb{C}}
\newcommand \system {\mathbb{S}}

\newcommand \head {H}

\newcommand* \stypeint {\sym{int}}
\newcommand* \Prop{\stypebool}

\newcommand* \stypebool {\mathit{bool}}

\newcommand* \wt[3][]{\ifthenelse{\isempty{#1}}{\senv}{\senv, #1} \vdash_H #2: #3}

\newcommand {\model} {\mathcal{M}}
\newcommand {\val}[1]{ \llbracket #1 \rrbracket }
\newcommand {\vwt}[3][]{\val{\wt{#2}{#3}}(\alpha\ifthenelse{\isempty{#1}}{}{[#1]})}

\tikzstyle{item} = [ellipse, rounded corners, minimum width=3cm, minimum height=1cm,text centered, draw=black]
\tikzstyle{process} = [rectangle, minimum width=2cm, minimum height=1cm, text centered, draw=black, fill=blue!30]
\tikzstyle{decision} = [rectangle, minimum width=2cm, minimum height=1cm, text centered, draw=black, fill=green!30]
\tikzstyle{arrow} = [thick,->,>=stealth]
\tikzstyle{result} = [diamond, minimum width=3cm, minimum height=1cm, text centered, draw=black, fill=green!30]

\lstdefinestyle{caml}{
 language=caml,
 basicstyle={\footnotesize\ttfamily},
 identifierstyle={\small},
commentstyle={\small\ttfamily \color[rgb]{0,0,0.5}},
keywordstyle={\small\bfseries \color[rgb]{0.8,0,0}},
keywordstyle = [2]{\color{orange}},
ndkeywordstyle={\small},
stringstyle={\small\ttfamily \color[rgb]{0,0,1}},
frame={tb},
breaklines=true,
columns=[l]{fullflexible},
numberstyle={\scriptsize},
stepnumber=1,
morecomment=[l]{//},
morekeywords={assert}
}

\floatname{algorithm}{Procedure}

%% file: intro2.tex
\section{Introduction}
\label{sec:intro}

Fully automated verification of programs through satisfiability checking of constrained Horn clauses (CHCs) has attracted considerable attention in recent years, as it offers a uniform and language-agnostic framework for verifying diverse program properties~\cite{seahorn,jayhorn,rusthorn,AltBHS22}.
For example, consider the following functional program.
\begin{verbatim}
let rec plus m n = if n=0 then m else (plus m (n-1))+1
let main m n = if n>=0 then assert(plus m n >= m)
\end{verbatim}
\newcommand\plus{\mathit{Plus}}
The lack of assertion failures in the above program can be reduced to the satisfiability of the following CHCs,
i.e, the problem of whether there exists a predicate \(\plus\) that satisfies them:
\begin{align*}
&  \forall m.\; \plus(m,0,m).\\
&  \forall m, n, r.\;\plus(m,n,r+1) \Leftarrow \plus(m, n-1, r).\\
  &  \forall m,n,r.\; \HFalse\Leftarrow \plus(m, n, r) \land n\ge 0 \land r<m.
\end{align*}
Here, \(\HFalse\) represents false.
The predicate \(\plus\) may be considered an invariant among arguments \(m, n\) and the corresponding return value \(r\);
indeed, \(\plus(m,n,r)\equiv m+n=r\) satisfies the above CHCs.
State-of-the-art CHC solvers~\cite{spacer,eldarica,gpdr,ChampionCKS20,BlichaFHS22} can quickly solve problems like
the above, finding an appropriate invariant (\(\plus(m,n,r)\equiv m+n=r\) or \(\plus(m,n,r)\equiv %
r\ge m\) in this case),
enabling fully automated program verification.

Despite various efforts, however, the current CHC solvers are not very good at dealing with data structures.
For example, consider the following variant of the example above, where natural numbers are represented as data structures.
\newcommand\lt{\mathit{Lt}}
\newcommand\plusnat{\mathit{PlusNat}}
\begin{align*}
&  \forall m.\; \plusnat(m,Z,m).\\
  &  \forall m, n, r.\;\plusnat(m,S(n),S(r)) \Leftarrow \plusnat(m, n, r).\\
  & \forall n.\;\lt(Z, S(n)).\qquad \forall m,n.\;\lt(S(m), S(n)) \Leftarrow \lt(m,n).\\
&  \forall m,n,r.\;\HFalse \Leftarrow \plusnat(m, n, r)\land \lt(r, m).
\end{align*}
Z3 Spacer~\cite{spacer}, a state-of-the-art CHC solver, fails to prove the satisfiability of the CHCs above.
A problem is that while a model for \(\plus\) can be expressed by a simple linear arithmetic formula (\(m+n=r\)),
an inductively defined predicate is required to express a model for \(\plusnat\), and it is, in general, hard to automatically
find such an inductively defined predicate and check that it is indeed a model (i.e., satisfies all the clauses).

To address the issue above, we propose a method for abstracting CHCs by using a \emph{catamorphism}~\cite{MeijerFP91} from
data structures to (tuples of) integers, on which existing CHC solvers perform well in practice.
For the example above, we can abstract \(Z\) and \(S\) respectively to
\(0\) and \(\lambda x.x+1\)\footnote{Actually, we do not lose any information using this \emph{abstraction}, as the induced catamorphism is injective.
In general, however a catamorphism may not be injective, hence introducing abstraction.},
and obtain the following ``abstract'' CHCs over integers, whose satisfiability implies that of the original CHCs.
\begin{align*}
&  \forall m.\; \AVAR{\plusnat}(m,0,m).\\
  &  \forall m, n, r.\;\AVAR{\plusnat}(m,n+1,r+1) \Leftarrow \AVAR{\plusnat}(m, n, r).\\
  & \forall n.\;\AVAR{\lt}(0, n+1).\qquad \forall m,n.\;\AVAR{\lt}(m+1, n+1) \Leftarrow \AVAR{\lt}(m,n).\\
&  \forall m,n,r.\;\HFalse \Leftarrow \AVAR{\plusnat}(m, n, r)\land \AVAR{\lt}(r, m).
\end{align*}
Universal quantifiers are now over integers, and for the sake of simplicity, we have omitted some conditions on the variables; %
see later sections.
State-of-the-art solvers can instantly deduce its satisfiability, which also implies the satisfiability of the original CHCs.

The remaining question is how to automatically find appropriate catamorphisms. In the above example, the required catamorphism
is just the ``size'' function incorporated by default in some CHC solvers like Eldarica~\cite{eldarica}, but as we will see later,
the size function is not always sufficient.
To this end, we propose a method for automatically finding appropriate catamorphisms in a counterexample-guided manner.
We have implemented the proposed method and developed a new CHC solver called \catalia{}.
According to our experiments using the benchmark set of the ADT category of CHC-COMP 2024,
\catalia{} significantly outperformed state-of-the-art CHC solvers for SAT instances.

Our contributions are summarized as follows.
\begin{itemize}
\item Formalization of abstraction of CHCs using catamorphisms.
\item Counterexample-guided automatic synthesis of appropriate catamorphisms.
\item  Implementation of the proposed method and experimental evaluation.
\end{itemize}
The idea of using catamorphisms for verification of programs with ADTs itself is not new~\cite{SuterDK10,PhamGW16,MukaiKS22,KSG22,AngelisFPP25}. Thus, our main contributions lie in the formalization and implementation of the procedure for
automatically discovering catamorphisms in the context of CHC solving.

The rest of this paper is organized as follows. Section~\ref{sec:pre} reviews CHCs over ADTs and catamorphisms for ADTs.
Section~\ref{sec:overview} gives an overview of our framework.
Section~\ref{sec:catalia-abstraction} formalizes our catamorphism-based abstraction, and Section~\ref{sec:cata-synthesis} introduces our template-based catamorphism synthesis procedure.
Section~\ref{sec:cata-eval} reports the experimental results.
Section~\ref{sec:rel} discusses related work, and Section~\ref{sec:conc} concludes the paper.

%% file: pre.tex
\section{Preliminaries}

\label{sec:pre}

\subsection{Constrained Horn Clauses modulo Algebraic Data Types and Integer Arithmetic}

We consider a standard first-order logic and a theory of algebraic data types and integer arithmetic, which is written as \( \TADTZ \).
We also denote the theory of integer arithmetic by \( \TZ \).
For simplicity, we consider formulas that involve only integers and a single algebraic data type (ADT) \( (\dsort, \{\, \CONS_1, \dots, \CONS_k \,\}) \) where \( \dsort \) is a sort and \( C_i \) is a function symbol called a \emph{constructor}.\footnote{Extending our proposed method to support multiple ADTs is straightforward. As detailed in Section~\ref{sec:cata-eval}, our implementation is already capable of handling multiple ADTs within a single instance.
  In a standard theory of algebraic data types, projections and testers are used as the standard connectives for algebraic data types in addition to constructors.
We omit them for simplicity,
since they can be easily removed by standard preprocessing techniques (c.f.\ Section 4.5 of \cite{KostyukovMF21}).
}
Each constructor \( \CONS_i \) is assumed to be a function of \( m_i + n_i \) arguments with its sort specified as
\[
    (\overbrace{\stypeint \times \dots \times \stypeint}^{m_i} \times \overbrace{\dsort \times \dots \times \dsort}^{n_i}) \to \dsort.
\]
Note that \( m_i + n_i \) can be \( 0 \).
We often simply write \( \dsort \) to mention the ADT.

\begin{example}
    An algebraic data type for natural numbers \( \sym{nat} \) in the introduction is defined as \( (\sym{nat}, \{\, Z, S \,\}) \) where \( Z \) is a constant of sort \( \sym{nat}\) and \( S \) has the sort \( \sym{nat} \to \sym{nat}\).
    An algebraic data type for integer lists, written as \( \sym{ilist} \), is defined as \( (\sym{ilist}, \{\, \sym{nil}, \sym{cons} \,\}) \) where \( \sym{nil} \) is a constant of sort \( \sym{ilist}\) and \( \sym{cons} \) has the sort \( \sym \stypeint \times \sym{ilist} \to \sym{ilist}\).
    A list \( [1; 2] \) is written as \( \sym{cons}(1, \sym{cons}(2, \sym{nil})) \).
\end{example}

The sets of \emph{terms} and \emph{constraint formulas} %
are defined by:
\begin{align*}
    (\textit{terms}) \quad \dterm &\BNFDef x \mid n \mid \CONS_i(\seq{\dterm_i}) \mid \dterm_1 \operator \dterm_2 \\
    (\textit{constraint formulas}) \quad \constraint &\BNFDef \true \mid \false \mid \constraint_1 \wedge \constraint_2 \mid \constraint_1 \vee \constraint_2 \mid t_1 \bowtie t_2 \mid \exists x^{\sort}.\: \constraint \mid \forall x^{\sort}.\: \constraint.
\end{align*}
Here, \( x \), \( n \), and \( \sort \) are metavariables for variables, integers, and sorts respectively, \( \bowtie \) ranges over binary predicates in \( \{\, =_{\stypeint}, \neq_{\stypeint}, >, \leq,  =_{\dsort} \,\}  \) and \( \operator \) ranges over binary arithmetic operations \( \{\, +, -, \times \, \} \).
We also write \( \seq{\dterm_i} \) for a sequence of terms \( \dterm_1, \dots, \dterm_k \).
The predicate \( =_{\dsort} \) takes two terms of sort \( \dsort \), while others take two terms of sort \( \stypeint \).
For technical convenience, we omit the disequality \(\neq_{\dsort}\) for ADT;
it can be encoded using CHCs. (c.f. Section 4.4 of \cite{KostyukovMF21}).
We also use \( (\dterm_1, \dots, \dterm_k) =_{\stypeint^k} (\dterm'_1, \dots, \dterm'_k) \) as a syntax sugar for the conjunction of equalities \( \dterm_1 =_{\stypeint} \dterm'_1 \land \dots \land \dterm_k =_{\stypeint} \dterm'_k \).

We write \( \fv(\dterm) \) and \( \fv(\constraint) \) for the set of free variables, and \( \forall \constraint \) and \( \exists \constraint \) for the universal and existential closures of \(\constraint\), respectively.
We also write \(  \universe{\herbrand}_{\dsort} \) for the set of ground terms (i.e., with no free variables) of sort \( \dsort \).
We call non-ground terms \emph{open terms}, and terms without constructors \emph{arithmetic terms}.
A constraint formula \( \constraint \) is \emph{quantifier-free} if no quantifier occurs in \( \constraint \).
We consider only well-sorted terms and formulas, where well-sortedness is defined in the standard manner.
\emph{Substitutions} such as \( [\dterm/x]\dterm' \) and \( [\dterm/x]\constraint \) are defined as usual.
The semantics of terms and formulas are also given in the standard way.

A \emph{constrained Horn clause} (CHC) (over algebraic data types and integer arithmetic) is a formula of the form
\begin{align*}
    \forall \seq{x^{\stypeint}}, \seq{y^{\dsort}}.\: \head \Leftarrow \constraint \land P_1(\seq{t_1}) \land \cdots \land P_k(\seq{t_k}).
\end{align*}
Here, \( P_i(\seq{t_i}) \) is a predicate application, \( \constraint \) is a quantifier-free constraint formula, and  \( \head \) is either \( \HFalse \) or a predicate application \( P(\seq{t_{k+1}}) \).
For simplicity, universal quantifiers are often omitted.
We call a finite set of CHCs \emph{a system of CHCs}.
We use \( \clause \) and \( \system \) as the metavariables for CHCs and systems of CHCs, respectively.
A system of CHCs is (or simply, CHCs are) said to be \emph{satisfiable} if there is an interpretation of predicate variables that makes all the clauses valid.

\newcommand \size {\mathit{size}}
\begin{example}
    \label{ex:plus-nat}
    Recall the CHCs over the ADT \(\sym{nat}\) in Section~\ref{sec:intro}, consisting of the predicates \( \plusnat \) and \( \lt \).
    The CHCs are satisfiable under the %
    model:
    \begin{align*}
        \plusnat(x, y, z) &\Def \size(x) + \size(y) = \size(z)
        \quad %
        \lt(x, y) \Def \size(x) < \size(y),
    \end{align*}
    where the function \(\size\) from \(\sym{nat}\) to integers is defined by \( \size(S^n(Z))=n \).
\end{example}

\subsection{Catamorphisms}

We introduce \emph{catamorphisms} for \( \dsort \), which are generalized fold functions that map instances of \( \universe{\herbrand}_{\dsort} \) to \( \NAPPROX \)-tuples of integers.
Here, \( \NAPPROX \) is called the \emph{approximation degree}. %
A catamorphism \( \cata \) (for \( \dsort \)) is defined as a map constructed as follows:
\begin{align*}
    \cata(x)\Def
    &\textbf{match \(x\) with} \\
    &\qquad
        \textbf{case } \CONS_1(\seq{y_1}, z_1, \dots, z_{n_1}) \Rightarrow
        \combine_1(\seq{y_1}, \cata(z_1), \dots, \cata(z_{n_1})) \\
    &\qquad \qquad \vdots \\
    &\qquad
        \textbf{case } \CONS_k(\seq{y_k}, z_1, \dots, z_{n_k}) \Rightarrow
        \combine_k(\seq{y_k}, \cata(z_1), \dots, \cata(z_{n_k})).
\end{align*}
Here, \( \combine_i \) is called a \emph{structure map} for \( \CONS_i \).
Recall that the constructor \( \CONS_i \) takes \( m_i \) arguments of sort \( \stypeint \) and \( n_i \) arguments of sort \( \dsort \).
The structure map \( \combine_i \) takes \( m_i \) integers %
and \( n_i \) \( \NAPPROX \)-tuples of integers, and returns an \( \NAPPROX \)-tuple of integers.
We also write \( \banana{\combine_1, \dots, \combine_k} \)
for \( \cata \) when clarifying the structure maps.
Furthermore, if a catamorphism has free variables \( a_1, \dots, a_l \) in the definition, we write \( \bananaT{\combine_1, \dots, \combine_k}{\set{a_1, \dots, a_l}} \).

\begin{example}
    \label{ex:cata}
    A catamorphism for \( \sym{ilist} \) is of the form
    \begin{align*}
      \banana{\combine_{\sym{nil}},\combine_{\sym{cons}}}(x)\Def
        &\textbf{match \(x\) with} \\
        &\qquad
            \textbf{case } \sym{nil} \Rightarrow \combine_{\sym{nil}} \\
        &\qquad
            \textbf{case } \sym{cons}(x, l) \Rightarrow \combine_{\sym{cons}}(x, \banana{\combine_{\sym{nil}},\combine_{\sym{cons}}}(l))
    \end{align*}
    where \( \combine_{\sym{nil}} \) and \( \combine_{\sym{cons}} \) are structure maps for the two constructors \( \sym{nil} \) and \( \sym{cons}\), respectively.
    For \(\combine_{\sym{nil}}=0\) and \(\combine_{\sym{cons}}(x, l) = 1+l\),
    \(\banana{\combine_{\sym{nil}},\combine_{\sym{cons}}}\) is the list length function,
    and for \(\combine_{\sym{nil}}=0\) and \(\combine_{\sym{cons}}(x, l) = x+l\),
    \(\banana{\combine_{\sym{nil}},\combine_{\sym{cons}}}\) is the function
    for computing the sum of list elements.
    The catamorphism defined by the structure maps
    \(\combine_{\sym{nil}}=(0,0)\) and \(\combine_{\sym{cons}}(x, (l_1,l_2)) = (1+l_1,x+l_2)\) has
    the approximation degree \(2\); it maps an integer list to a pair consisting
    of  the list length and the sum of elements.

    Similarly, we define a catamorphism \( \cata_{\size} \) for \( \sym{nat} \) as \( \banana{\combine_{Z}, \combine_{S}} \) where \( \combine_{Z} = 0 \) and \( \combine_{S}(x) = 1 + x \). This corresponds to \( \size \) used in Example~\ref{ex:plus-nat}.
\end{example}

%% file: overview3.tex
\section{Overview}
\label{sec:overview}

This section gives an overview of the proposed procedure, called \catalia{}, which follows a framework of template-based synthesis and counterexample-guided abstraction refinement (CEGAR)~\cite{ClarkeGJLV03}, as illustrated in Figure~\ref{fig:cata-overview}.

Given a system \(\system\) of CHCs over algebraic data types and integer arithmetic, %
we first abstract them to a system \(\system'\) of CHCs over integer arithmetic. %
When \(\system'\) is unsatisfiable, \catalia{} generates a constraint formula \( \constraint \) (called a counterexample),
which witnesses the possible unsatisfiability of \(\system\), based on a resolution proof for unsatisfiability of \(\system'\).
If \( \constraint \) is satisfiable, \(\system\) is indeed unsatisfiable; otherwise, we refine the catamorphism \( \cata \) using \( \constraint \) in the refinement phase.

\begin{figure}[t]
    \begin{tikzpicture}[node distance=2cm]

      \node (chc) [item] {CHC over \( \TADTZ \)};
      \node (cataA) [process, below of=chc,yshift=0.5cm,text width=3cm] {Abstraction \( \cataA \) \\ (Section~\ref{sec:catalia-abstraction})};
      \node (abstracted) [item, below of=cataA,yshift=0.5cm] {CHC over \( \TZ \)};
      \node (isSat) [decision, below of=abstracted,yshift=0.5cm] {Is SAT?};
      \node (sat) [left of=isSat, xshift=-0.4cm] {SAT};
      \node (cex) [item, right of=isSat, xshift=2cm] {CEX \( \constraint \) };
      \node (isfeasible) [decision, right of=cex, yshift=1.2cm,xshift=0.2cm] {Is SAT?};
      \node (unsat) [right of=isfeasible, xshift=0.4cm] {UNSAT};
      \node (synthesis) [process, above of=isfeasible, text width=3cm,xshift=1.8cm,yshift=-0.2cm] {Refinement \\ (Section~\ref{sec:cata-synthesis})};
      \node (cata) [item, left of=synthesis, xshift=-2cm] {Catamorphism \( \cata \)};

      \draw [arrow] (chc) -- (cataA);
      \draw [arrow] (cataA) -- (abstracted);
      \draw [arrow] (abstracted) -- (isSat);
      \draw [arrow] (isSat) -- node[anchor=south] {Yes} (sat);
      \draw [arrow] (isSat) -- node[anchor=south] {No} (cex);
      \draw [arrow] (cex) -- (isfeasible);
      \draw [arrow] (isfeasible) -- node[anchor=south] {Yes} (unsat);
      \draw [arrow] (isfeasible) -- node[anchor=east] {No} (synthesis);
      \draw [arrow] (synthesis) -- (cata);
      \draw [arrow] (cata) -- (cataA);

      \end{tikzpicture}

    \caption{Overview of \catalia{}}
    \label{fig:cata-overview}
  \end{figure}

Below, we briefly explain the abstraction, counterexample generation, and synthesis phases of \catalia{}.
We will provide more details for the abstraction and synthesis phases in Sections~\ref{sec:catalia-abstraction} and \ref{sec:refinement}, respectively.

\subsection{Abstraction}

\label{sec:abstraction-overview}

Recall the following system of CHCs given in Section~\ref{sec:intro}.
\begin{align*}
(i)\ &  \plusnat(m,Z,m).\\
(ii)\ &    \plusnat(m,S(n),S(r)) \Leftarrow \plusnat(m, n, r).\\
(iii)\ & \lt(Z, S(n)).\qquad
(iv)\ \lt(S(m), S(n)) \Leftarrow \lt(m,n).\\
(v)\ &  \false \Leftarrow \plusnat(m, n, r)\land \lt(r, m).
\end{align*}

Using the catamorphism \(\banana{\combine_{Z}, \combine_{S}} \) where \( \combine_{Z} = 0 \) and
\( \combine_{S} (x) =  1 + x \), we obtain the following abstracted version of CHCs.
\begin{align*}
&  \AVAR{\plusnat}(m,0,m)\Leftarrow \CODP(m).\\
 &   \AVAR{\plusnat}(m,n+1,r+1) \Leftarrow \CODP(m)\land \CODP(n)\land \CODP(r)\land \AVAR{\plusnat}(m, n, r).\\
  & \AVAR{\lt}(0, n+1)\Leftarrow \CODP(n).\qquad \AVAR{\lt}(m+1, n+1) \Leftarrow \CODP(m)\land \CODP(n)\land \AVAR{\lt}(m,n).\\
  &  \false \Leftarrow \CODP(m)\land \CODP(n)\land \CODP(r)\land \AVAR{\plusnat}(m, n, r)\land \AVAR{\lt}(r, m).\\
  & \CODP(0). \qquad \CODP(n+1) \Leftarrow \CODP(n).
\end{align*}
Here, we have replaced each variable or term of sort \(\nat\) with one of sort \(\stypeint\) by applying the catamorphism.
We have also added \(\CODP(x)\) for each universally quantified variable \(x\) (whose sort is \(\nat\) in the original
CHCs and \(\stypeint\) after the abstraction); this is for the purpose of
restricting the range of the variable to the image of the catamorphism.
(In other words, a formula \(\forall x^\nat.\varphi(x)\) is abstracted to \(\forall x^\stypeint.\CODP(x)\Rightarrow \varphi'(x)\)
where \(\varphi'(x)\) is an abstract version of \(\varphi(x)\).)
The definition of \(\CODP\), called the \emph{\(\cata\)-admissibility predicate}, is obtained automatically from the following predicate that should be
satisfied by every variable of sort \(\nat\), which was implicit in the original CHCs.
\begin{align*}
&  \CODPo(Z). \qquad \CODPo(S(x)) \Leftarrow \CODPo(x).
\end{align*}
Note that, without the predicate \(P\), the abstraction would be too coarse.
For example, consider a CHC \( \forall x^{\nat}.\: \false \Leftarrow S(x) = Z\), which is valid (as \(S(x)=Z\) never holds).
Without the \(P\) predicate, however, it would be abstracted to
\( \forall x^{\stypeint}.\: \false \Leftarrow x+1=0\), which is invalid.

The abstracted CHCs above are satisfiable with the following model:
\begin{align*}
  \AVAR{\plusnat}(m,n,r) \Def m + n = r \qquad
  \AVAR{\lt}(x, y) \Def x < y \qquad
  \CODP(n) \Def n \geq 0.
\end{align*}
We can therefore conclude that the original CHCs are also satisfiable, with the following model:
\begin{align*}
  \plusnat(m,n,r) \Def \cata(m) + \cata(n) = \cata(r) \qquad
  \lt(x, y) \Def \cata(x) < \cata(y).
\end{align*}

A remaining issue is how to abstract the primitive equality predicate \(=_{\nat}\). We can simply replace it with \(=_{\stypeint}\);
this is a sound abstraction, since \(x=_\nat y\) implies \(\cata(x)=_{\stypeint}\cata(y)\). In contrast, \(x\ne_\nat y\) does NOT
imply \(\cata(x)\ne_{\stypeint}\cata(y)\); that is why we exclude out \(\ne_\delta\) from the set of primitive predicates
(recall Section~\ref{sec:pre}), and encode the inequality by using CHCs.

\subsection{Counterexample Generation}
\label{sec:cex-gen-overview}

Now, we consider the case where the abstracted CHCs are unsatisfiable.
If the catamorphism \(\banana{0, \lambda x.0}\) (which maps all the natural numbers to \(0\))
were used instead of \(\banana{0, \lambda x.x+1}\),
the original CHCs would be abstracted to the following CHCs over \( \TZ \), which are unsatisfiable.
\begin{align*}
    (i')\ &  \AVAR{\plusnat}(m,0,m) \Leftarrow \CODP(m).\\
  (ii')\ & \AVAR{\plusnat}(m,0,0) \Leftarrow \CODP(m) \land \CODP(n) \land \CODP(r) \land \AVAR{\plusnat}(m, n, r).\\
  (iii')\ & \AVAR{\lt}(0, 0)\Leftarrow \CODP(n). \qquad\
  (iv') \ \AVAR{\lt}(0, 0) \Leftarrow \CODP(m) \land \CODP(n) \land \AVAR{\lt}(m,n).\\
 (v') \ & \false \Leftarrow \CODP(m) \land \CODP(n) \land \CODP(r) \land \AVAR{\plusnat}(m, n, r)\land \AVAR{\lt}(r, m).\\
    (vi')\ & \CODP(0). \qquad (vii')\ \CODP(0)\Leftarrow \CODP(m).
\end{align*}
Notice that the abstracted CHCs \( (i') - (v') \) correspond to the original CHCs \( (i) - (v) \).

Suppose the following (SLD-)resolution proof was generated by a CHC solver as a witness of unsatisfiability.
\[
\infer[(vi')]{\HFalse\Leftarrow
  m=n=r=n'=0}
{\infer[(iii')]{\HFalse\Leftarrow
    \CODP(m) \land \CODP(n) \land \CODP(r) \land n = 0 \land m=r\land \CODP(m) \land m=r = 0  \land \CODP(n')}
  {\infer[(i')]{\HFalse\Leftarrow
      \CODP(m) \land \CODP(n) \land \CODP(r) \land n = 0 \land m=r\land \CODP(m) \land \AVAR{\lt}(r, m)}
    {\HFalse\Leftarrow \CODP(m) \land \CODP(n) \land \CODP(r) \land \AVAR{\plusnat}(m, n, r)\land \AVAR{\lt}(r, m)}}}
\]
The derivation starts with the goal clause, and
in the last step, the resolution on multiple occurrences of \(\CODP\) has been performed in one step.
We have indicated which CHC has been used in each resolution step.
Note that the last clause is invalid, which indicates that the abstracted CHCs are unsatisfiable.

From the resolution proof for the abstract CHCs above,
we construct the following candidate of a resolution proof for the unsatisfiability of the original CHCs,
by applying the corresponding clause
of the original CHCs (i.e., the clauses (i) and (iii) instead of (i') and (iii') respectively) except for the clauses on
the \( \cata \)-admissibility predicate.
\[
\infer[(iii)]{\HFalse\Leftarrow
  m = r \land n = Z \land r = Z \land m = S(n')}
{\infer[(i)]{\HFalse\Leftarrow
    m = r \land n = Z \land \lt(r, m)}
  {\HFalse\Leftarrow \plusnat(m, n, r) \land \lt(r, m)}}
\]
The right-hand side of the last clause can be simplified to:
\(
\constraint \Def %
Z = S(n')
\).
We call \( \constraint \) a \emph{counterexample} (against the satisfiability of the original CHCs);
it serves as a possible witness of the unsatisfiability of the original CHCs,
in the sense that if \(\constraint\) were satisfiable, we could conclude that the original CHCs were unsatisfiable.
In this case, however, \( \constraint \) is unsatisfiable;
the abstract CHCs yielded a spurious resolution proof for the original CHCs.

\subsection{Synthesis}
\label{sec:refinement-overview}

For the example in the previous subsection (where \(\theta \equiv Z = S(n')\)), we need to find a catamorphism \(\cata'=\banana{\combine_Z,\combine_S}\) such that
\[\exists n^\nat.\, \combine_Z=\combine_S(\cata'(n))\]
is invalid.
By preparing a template
\(
  \combine_{Z} = a \mbox{ and }\combine_{S} (\AVAR{n}) =  b \times \AVAR{n} + c
  \)
  for \(\combine_{Z}\) and \(\combine_{S}\),
  the problem above is reduced to the satisfiability problem:
  \[
\exists a,b,c.\,  \forall n^\nat.\,a \ne b\times\banana{\combine_Z,\combine_S}(n) + c.
  \]

  Solving the above satisfiability problem is costly, since it is an \(\exists\forall\)-formula and also involves the
  recursive function \(\banana{\combine_Z,\combine_S}\).
We thus provide a procedure for solving this \(\exists\forall\)-formula based on counterexample-guided inductive synthesis (CEGIS), which we will explain in detail in Section~\ref{sec:refinement}.
This procedure may return \( a = 0 \) and \( b = c = 1 \) as a witness and synthesize a new catamorphism \(\banana{0, \lambda x.x+1}\).
We then go back to the abstraction step of the CEGAR cycle, and in this case, succeed in proving the satisfiability of
the original CHCs, as explained in Section~\ref{sec:abstraction-overview}.

%% file: catalia2.tex
\section{Abstraction}
\label{sec:catalia-abstraction}

This section explains more details about the abstraction step.

In this section,
we use the following more tricky example as a running example.
To the best of our knowledge,
most of the previous approaches~\cite{eldarica,spacer,KostyukovMF21,KSG22} struggle with the example,%
while our approach can easily solve it.

\begin{example}
    \label{ex:eq-odd-even-sum}
    We consider the following CHCs:
\begin{align*}
    & G(\nil, 0, 0).\qquad
 G(\cons(x, l), x + n, m) \Leftarrow G(l, m, n).\\
    & \sym{Gen}(\nil, 0). \qquad
     \sym{Gen}(\cons(x, \cons(x - 1, l)), n) \Leftarrow \sym{Gen}(l, n - 1). \\
    & \HFalse \Leftarrow m - n\ne x \land x \geq 0 \wedge \sym{Gen}(l, x) \wedge G(l, m, n).
\end{align*}
Here, \( G \) and \(\sym{Gen}\) are predicate symbols of sorts \( \ilist \times \stypeint \times \stypeint \to \stypebool \), and  \( \ilist \times \stypeint \to \stypebool \) respectively.
Intuitively, \(G(l,m,n)\) means that the sums of elements of \(l\) at even and odd indices
are \(m\) and \(n\) respectively (where an index starts from \(0\)).
The predicate \(\sym{Gen}(l,n)\) holds if \(l\) is a list of the form
\([m_0;m_0';\cdots;m_{n-1};m'_{n-1}]\) where \(m_i = m'_i+1\) for each \(i=0,\ldots,n-1\).
The last clause asserts that if \(\sym{Gen}(l,x)\) holds, then
the difference between the sums of elements at even and odd indices is \(x\).
The above system of CHCs is satisfiable, where the models of \(G\) and \(\sym{Gen}\) are
as informally explained above.

Let \( \cataeo \) be \( \banana{\combine_{\nil}, \combine_{\cons}} \) where \( \combine_{\sym{nil}} = 0 \) and \( \combine_{\sym{cons}}(x, l) = x - l \).
This catamorphism is sufficient for proving the satisfiability of the CHCs above.
\end{example}

Suppose the approximation degree is \( N \), and we have a catamorphism \( \banana{\combine_1, \dots, \combine_k} \), which is a map from \( \universe{\herbrand}_{\dsort} \) to \( \Int^{N} \).
We define an abstraction of CHCs \( \cataA \) induced by \( \cata \).

Before applying the abstraction, we add
the atom \(\CODPo(x)\), where \(\CODPo \) is a unary predicate that takes a value of sort \( \dsort \), to the body of each clause, for each variable \(x\) of sort \(\delta\).
Intuitively, \(\CODPo(x)\) means that \(x\) ranges over the set of
terms of sort \(\delta\). We add the following clause for
each constructor \(\CONS_i\) of sort
\(
    ( \overbrace{\stypeint \times \dots \times \stypeint}^{m_i} \times \overbrace{\dsort \times \dots \times \dsort}^{n_i} ) \to \dsort
\).
\[ \CODPo(\CONS_i(x_1,\ldots,x_{m_i},y_1,\ldots,y_{n_i}))\Leftarrow
\CODPo(y_1)\land \cdots \land \CODPo(y_{n_i}).\]
It ensures that the least model of \(\CODPo\) indeed has the meaning described above.
Obviously, the CHCs augmented with \(\CODPo\) is equi-satisfiable with the
original CHCs.
Nonetheless, \( \CODPo \) is added because the abstraction of \(\CODPo\) yields the ``\(\cata\)-admissibility predicate'' mentioned
in Section~\ref{sec:overview}.
The abstraction of \( \CODPo(x) \) ensures the integer variables obtained by abstracting \( x \) to range over the image of the catamorphism.

\begin{example}
\label{ex:augmented}
  For the example in Example~\ref{ex:eq-odd-even-sum}, the augmented CHCs are:
  \begin{align*}
    (i)\ &  G(\nil, 0, 0).\qquad
 (ii)\ G(\cons(x, l), x + n, m) \Leftarrow P_\ilist(l)\land G(l, m, n).\\
 (iii)\ &  \sym{Gen}(\nil, 0). \\
     (iv)\  & \sym{Gen}(\cons(x, \cons(x - 1, l)), n) \Leftarrow P_\ilist(l)\land \sym{Gen}(l, n - 1). \\
     (v)\ &  \HFalse \Leftarrow m - n\ne x \land x \geq 0 \wedge P_\ilist(l)\land
      \sym{Gen}(l, x) \wedge G(l, m, n).\\
      (vi)\ & P_\ilist(\nil).\qquad
      (vii)\  P_\ilist(\cons(x, l))\Leftarrow P_\ilist(l).
\end{align*}
\end{example}
The augmented CHCs are abstracted as follows.
Let \( \clause \) be a clause \( \forall \seq{x^{\stypeint}}, \seq{y}^{\dsort}.\: \head \Longleftarrow \constraint \wedge P_1(\seq{t_1}) \land \dots \land P_n(\seq{t_n}) \), and
\( \CENV \) be a variable abstraction environment, which is a finite map from variables to \( N \)-tuples of variables, such that
\(\CENV(y_i)=(y^i_1, \dots, y^i_\NAPPROX)  \) for each \(y_i\in \set{\seq{y}^\dsort}=\set{y^1,\ldots,y^l}\).
We define \( \cataA(\clause) \) as
\begin{align*}
    \forall \seq{x}, \seq{y}^{i}_{j}.\:
    \cataAFB(\head) &\Longleftarrow
    \cataAFB(\constraint) \land \cataAFB(P_1(\seq{t}_1)) \land \dots \land \cataAFB(P_n(\seq{t}_n)).
\end{align*}
Here, \(\seq{y}^{i}_j\) denotes
\(y^1_1, \dots, y^1_\NAPPROX,\ldots,y^l_1, \dots, y^l_\NAPPROX\),
and
the abstraction \(\cataAFB\) for atoms, constraint formulas, and terms is defined by:
\begin{align*}
  &  \cataAFB(P(\seq{t}))=\AVAR{P}(\cataAT(\seq{t}))\qquad
    \cataAFB(\true)  =  \true \qquad \qquad
    \cataAFB(\false)  =  \false
    \\
 &   \cataAFB(t_1 =_{\dsort} t_2) = \cataAT(\dterm_1) =_{\stypeint^{\NAPPROX}} \cataAT(\dterm_2)
    \\
  &  \cataAFB(\predicate(\arith_1, \dots, \arith_l)) = \predicate(\arith_1, \dots, \arith_l) \qquad \mbox{ (if \( \predicate \) is a built-in predicate on integers)}
    \\&
    \cataAFB(\form_1 \land \form_2) = \cataAFB(\form_1) \land \cataAFB(\form_2)
    \qquad  \cataAFB(\form_1 \lor \form_2) = \cataAFB(\form_1) \lor \cataAFB(\form_2)\\
    &    \cataAT(x) = \left\{\begin{array}{ll}
     \CENV(x) & \mbox{if \( x \in \dom(\CENV) \) }\\
    x &\mbox{otherwise} \end{array}\right.\qquad
    \cataAT(\arith)  = \arith  \mbox{ (if \( \arith \) is an integer term)}
    \\
   & \cataAT(\CONS_i(\arith_1, \dots, \arith_{m_i}, \dterm_1, \dots, \dterm_{n_i}))
      =  \combine_i(\arith_1, \dots, \arith_{m_i}, \cataAT(\dterm_1), \dots, \cataAT(\dterm_{n_i})).
\end{align*}
As defined above, we just recursively replace
each constructor \(\CONS_i\) with the corresponding structure map \(\combine_i\),
and the equality \(=_{\dsort}\) on ADT with the equality \(=_{\stypeint^{\NAPPROX}}\)
on integer tuples.

\begin{example}
  \label{ex:abstracted}
  Recall the augmented CHCs in Example~\ref{ex:augmented}
  and the catamorphism \(\cataeo\) given in Example~\ref{ex:eq-odd-even-sum}.
  We obtain the following abstracted CHCs.
  \begin{align*}
    & \AVAR{G}(0, 0, 0).\qquad
 \AVAR{G}(x-l, x + n, m) \Leftarrow \AVAR{P}_\ilist(l)\land \AVAR{G}(l, m, n).\\
    & \AVAR{\sym{Gen}}(0, 0). \qquad
     \AVAR{\sym{Gen}}(x-((x-1)-l), n) \Leftarrow \AVAR{P}_\ilist(l)\land \AVAR{\sym{Gen}}(l, n - 1). \\
     & \HFalse \Leftarrow m - n\ne x \land x \geq 0 \wedge \AVAR{P}_\ilist(l)\land
     \AVAR{\sym{Gen}}(l, x) \wedge \AVAR{G}(l, m, n).\\
     &  \AVAR{P}_\ilist(0).\qquad \AVAR{P}_\ilist(x-l)\Leftarrow \AVAR{P}_\ilist(l).
  \end{align*}
  They have the following model.
  \[
  \AVAR{G}(l,y,z) \Def l=y-z\qquad \AVAR{\sym{Gen}}(l,y)\Def l=y
  \qquad \AVAR{P}_\ilist(l) \Def \true.\]
  We have \( \AVAR{P}_\ilist(l) \Def \true \) because the image of \( \cataeo \) is \( \mathbb Z \), but we note that this is not the case in general.\footnote{For example, a model for the predicate \( \CODP \) in Section~\ref{sec:abstraction-overview} is given by \( \CODP(x) \Def x \ge 0 \).}

  By the soundness theorem given below, we can conclude that the original
  CHCs given in Example~\ref{ex:eq-odd-even-sum} are satisfiable.
Indeed, the following is the model for the original CHCs.
    \begin{align*}
        G(l, y, z) \Def \cataeo(l) = y - z
        \quad
        \mbox{and}
        \quad
        Gen(l, y) \Def \cataeo(l)  = y.
        \end{align*}
\end{example}

We now discuss the soundness of abstraction.
The following lemma follows immediately from the above construction.
\begin{lemma}
    For any system \( \system \) of CHCs, \( \cataA(\system) \) is a system of CHCs defined over integer arithmetic.
\end{lemma}
The following theorem states the soundness, and also describes how a model of the original CHCs can be constructed
from that of the abstract CHCs.
\begin{theorem}[Soundness]
    \label{th:cata-soundness}
    Let \( \system \) be a system of CHC over \( \TADTZ \).
    If \(\cataA(\system)\)  is satisfiable, then so is \( \system \).
    Furthermore, from a model \( \model_1\) of \( \cataA(\system) \), we can construct a model \( \model_2 \) for \( \system \).
    The interpretation of a predicate \( \PREDSYM \colon \dsort^m \times \stypeint^n \) in \( \model_2 \) is given as a function \( \PREDSYM^{\model_2}(x_1, \ldots, x_m, \seq{y}) = \AVAR{\PREDSYM}^{\model_1}(\cata(x_1), \ldots, \cata(x_m), \seq{y})\), where \( \AVAR{\PREDSYM}^{\model_1} \) is the interpretation of \( \AVAR{\PREDSYM} \) in \( \model_1 \).
\end{theorem}

A proof is given in \iflong Appendix~\ref{appx:soundness}.
\else the longer version of this paper~\cite{sas2025arxiv}. \fi
Here we provide an informal argument.
Notice that, for any predicate
\( \PREDSYM \colon \dsort^m \times \stypeint^n \), ground terms \(t_1,\ldots,t_m\) and integers \(\seq{k}\),
\(\PREDSYM(t_1,\ldots,t_m,\seq{k})\) holds under \(\model_2\) (i.e.,
\(\PREDSYM^{\model_2}(t_1, \ldots, t_m, \seq{k})\) holds) if and only if
\(\cataA(\PREDSYM(t_1,\ldots,t_m,\seq{k})) = \AVAR{\PREDSYM}(\cata(t_1),\ldots,\cata(t_m),\seq{k})\) holds under \(\model_1\)
(i.e., \(\AVAR{\PREDSYM}^{\model_1}(\cata(t_1),\ldots,\cata(t_m),\seq{k})\) holds).
For any ground constraint formula \(\theta\), \(\neg \cataA(\theta)\) implies \(\neg \theta\).
Thus, for any ground clause
\( \head \Leftarrow \constraint \land P_1(\seq{t_1}) \land \cdots \land P_k(\seq{t_k})\),
if \(\model_1\) is a model of
\( \cataA(\head) \Leftarrow \cataA(\constraint) \land \cataA(P_1(\seq{t_1})) \land \cdots \land \cataA({P}_k(\seq{t_k}))\),
then \(\model_2\) is a model of
\( \head \Leftarrow \constraint \land P_1(\seq{t_1}) \land \cdots \land P_k(\seq{t_k})\).

Recall the model of the original CHCs in Example~\ref{ex:abstracted}.
The interpretation
\( G^{\model_2}(l, x, y) \Def %
(\cataeo(l) = x - y) \)
for \(G\) has been obtained from
\( \AVAR{G}^{\model_1}(l, x, y) \Def (l = x - y) \), by just replacing \(l\) with \(\cataeo(l)\)
based on the theorem above.

%% file: synthesis3.tex
\section{Template-based and Counterexample-guided Catamorphism Synthesis}
\label{sec:cata-synthesis}
\label{sec:refinement}
\newcommand*{\param}{a}
\newcommand*{\paramTwo}{b}

We adopt a \emph{template-based approach} to catamorphism synthesis; we prepare a set of predefined template catamorphisms and derive constraints for \( \cata \) to satisfy.
As detailed later, such constraints involve universal quantifiers over ADTs and recursively defined functions, which are difficult for SMT solvers to handle.
To address this problem, we employ a \emph{counterexample-guided approach with testing}.

\subsection{Template-based Catamorphism Synthesis}

We define a \emph{template catamorphism} as a catamorphism parameterized by integers. %
Here, each constructor \( \CONS_i \) has an associated \emph{template structure map} represented by a tuple of open terms \( (t_i^{(1)}[\seq{x}, \seq{a}], \ldots, t_i^{(N)}[\seq{x}, \seq{a}]) \) of sort \( \stypeint \times \dots \times \stypeint \)
 where \( \seq{x} \) are arguments of the catamorphism and \( \seq{a} \) are integer parameters.

Let \( \cataTMP \) be a template catamorphism, and \( \model \) be an assignment of the parameters \( \seq{a} \) to integers.
A catamorphism \( \model(\cataTMP) \), a catamorphism obtained by substituting each parameter \( a_i \) with \( \model(a_i) \) in \( \cataTMP \),
is said to be an \emph{instantiation} of the template catamorphism \( \cataTMP \).

We often use \emph{linear} template catamorphisms, which are
 template catamorphisms whose associated template structure maps are affine functions.
We denote a linear template catamorphism of degree \(1\) by \(\linearTMP\).
\begin{example}
\label{ex:cata-linear-template-catamorphism}
For \(\sym{ilist}\), \(\linearTMP\) is given by the template structure maps:
\[
\combine_{\nil} = d, \quad \combine_{\cons}(x, l) = a \times l + b \times x + c
\]
with parameters \( a, b, c, d \).
\end{example}

\begin{remark}
    \label{rem:implementation-template}
    In practice, we design our templates based on a trade-off between the expressive power of the abstraction and the cost of catamorphism synthesis.
    As explained in Section~\ref{sec:cata-eval},
    our implementation uses restricted linear template catamorphisms \( \linearTMP_{[a, b]} \)
    to efficiently explore the search space, where
    the range of each parameter of \( \linearTMP \) is limited to \( [a, b] \).
    For example, a restricted linear template catamorphism for \( \ilist \), denoted by \( \linearTMP_{[-1, 1]} \), is given by:
    \[
        \combine_{\nil} = d, \quad \combine_{\cons}(x, l) = a \times l + b \times x + c
        \quad
        \mbox{where } a, b, c, d \in [-1, 1].
    \]

    We gradually increase the expressiveness of the templates by increasing the parameter ranges and the approximation degree, until they suffice to prove the satisfiability of the given CHCs.
    In our implementation, we prepare the following sequence of template catamorphisms:
    \begin{align*}
    &\linearTMP_{[-1, 1]},
    (\linearTMP_{[-1,1]}, \linearTMP_{[-1,1]}),
    (\linearTMP_{[-1,1]},\linearTMP_{[-1,1]},\linearTMP_{[-1,1]}),
    (\linearTMP_{[-2,2]}, \linearTMP_{[-2,2]}, \linearTMP_{[-2,2]}),
    \\
    &(\linearTMP_{[-4,4]}, \linearTMP_{[-4,4]}, \linearTMP_{[-4,4]}),\ldots.
    \end{align*}
    Here, we assume each \( \linearTMP_{[a, b]} \) has its own unique set of parameters, and
    an \(N\)-tuple of template catamorphisms \( (\cataTMP_1, \dots, \cataTMP_k)\) represents a template catamorphism of \( N \)-approximation degree, defined by \( \slambda \seq{x}.\, (\cataTMP_1(\seq{x}), \dots, \cataTMP_k(\seq{x})) \).

    We could also consider templates containing disjunctive properties, but inferring such templates would be more costly;
    it is left for future work.

    Note that the templates are prepared and fixed in advance, and are independent of specific CHCs or data structures.
    The parameters required for each constructor are determined by its sort, and can be derived automatically.
\end{remark}

We introduce a constraint generation map \( \encoderalt{\cdot}{\cataTMP} \),
which is used later to find appropriate instantiations of template catamorphisms. It is defined by:
\begin{align*}
  & \encoderalt{\forall x.\theta}{\cataTMP} = \forall x.\encoderalt{\theta}{\cataTMP}\quad
  \encoderalt{t_1 =_{\dsort} t_2}{\cataTMP} = (\cataTMP(t_1) =_{\stypeint^N} \cataTMP(t_2)) \quad
\encoderalt{\predicate(\seq{\arith})}{\cataTMP} = \predicate(\seq{\arith})\quad   \\
&\encoderalt{\theta_1 \star \theta_2}{\cataTMP} = \encoderalt{\theta_1}{\cataTMP} \star \encoderalt{\theta_2}{\cataTMP} \text{ for } \star \in \{\land, \lor\} \quad \encoderalt{\lnot \theta}{\cataTMP} = \lnot \encoderalt{\theta}{\cataTMP}.
\end{align*}
Here, we assume that \( \cataTMP \) is a symbol of a recursively defined function of sort \( \dsort \to \stypeint^{N} \) that encodes the template catamorphism.
Note that in Section~\ref{sec:refinement} we use an extended form of constraint formulas (used only for our catamorphism synthesis) that allows the negation operator \( \lnot \) for convenience.

\begin{example}
    \label{ex:cata-synth-recdef}
    Let \( \constraint \) be the (extended) constraint formula \( \lnot(\cons(0, \cons(0, l)) =_{\dsort} \nil) \), and \( \cataTMP \) be the linear template catamorphism \( \bananaT{\combine_{\sym{nil}},\combine_{\sym{cons}}}{\set{a, b, c, d}} \) in Example~\ref{ex:cata-linear-template-catamorphism}.
   By applying \( \encoderalt{\cdot}{\cataTMP} \) to \(  \constraint \), we obtain \(
   \lnot( \cataTMP(\cons(0, \cons(0, l))) =_{\stypeint} \cataTMP(\nil)) \),
     which can be simplified to
    \begin{align*}
       a \times (a \times \cataTMP(l) + c) + c \ne_{\stypeint} d
\end{align*}
    by using the defining axioms of \( \cataTMP \).
\end{example}

\subsection{Counterexample-Guided Catamorphism Synthesis}

\begin{algorithm}[t]
    \caption{Counterexample-Guided Catamorphism Synthesis}
    \label{alg:cegar-overview}

    \begin{algorithmic}[1]
        \State \textbf{Input: } CHCs over \( \TADTZ \) $\system$
        \State \textbf{Output: } \texttt{satisfiable} / \texttt{unsatisfiable} / \texttt{unknown}
        \State $(C, S) \gets (\defaultCATA, \emptyset)$
        \For{\( \cataTMP \) in \( \templates \) }
        \State $\cache \gets \true $
        \Loop
            \State $\system' \gets \cataA(\system)$
            \State $r \gets \text{check\_sat\_chc}(\system')$
            \State \textbf{if} $r = \texttt{satisfiable}$ \textbf{then return} \texttt{satisfiable}
            \State $\constraint \gets \text{get\_cex}(r, \system)$
            \State \textbf{if} $\text{check\_sat\_smt}(\constraint)$ \textbf{then return} \texttt{unsatisfiable}

            \State $S \gets S \cup \set{ \forall\lnot \constraint}$
            \State $(C, \cache) \gets \text{synthesis}(\bigwedge S, C, \cache, \cataTMP)$
            \State \textbf{if} $(C, \cache)$ is \texttt{None} \textbf{then} \textbf{break}
        \EndLoop
        \EndFor
        \State \textbf{return} \texttt{unknown}
    \end{algorithmic}
\end{algorithm}

We now discuss the CEGAR procedure of \catalia{} in more detail, shown in Procedure~\ref{fig:cata-overview}.

The procedure maintains two internal states:
the current catamorphism \( C \) and a set \( S \) of the negations of counterexample formulas.
We call elements of \(S\) \emph{proof obligations}; they are valid formulas over ADTs, whose validity
should be preserved by the catamorphism-based abstraction.
Initially, we set \( C \) to the default catamorphism \( \defaultCATA \) and \( S \) to the empty set.
The choice of \( \defaultCATA \) is arbitrary.

We iterate over a sequence of template catamorphisms \( \templates \), which is prepared in advance as described in Remark~\ref{rem:implementation-template}.
For now, let us ignore \( \cache \) and focus on the inner loop (line 6-15).
The first part of the inner loop (line 7-11) is the same as described in Section~\ref{sec:overview}.
When a candidate counterexample \( \constraint \) for \( \system \) is spurious, we add \( \forall \lnot \constraint \) to the set \( S \) and proceed to the synthesis phase (line 12-13).
When the synthesis phase fails to find a new %
catamorphism, we break the loop and try another template catamorphism (line 14).

A notable difference from the standard CEGAR approach is to relax the goal of the synthesis procedure:
to tackle the challenges described below,
we allow it to return a catamorphism \( \cata \) that does not necessarily preserve the validity
of the proof obligation \(\bigwedge S\).
As a result, the same spurious counterexample \( \constraint \) might be encountered multiple times at line 10.
To prevent it,
we store a constraint formula \( \cache \), which accumulates information from the synthesis exploration.
This formula represents necessary conditions for the template parameters,
enabling the synthesis process to resume from its previous state when needed, as detailed below.

\subsubsection{Challenges in Catamorphism Synthesis}

Synthesizing a catamorphism from the proof obligation \(\bigwedge S\)
faces two main challenges:
\begin{enumerate}
    \item The proof obligation involves universal quantifiers over ADTs and recursive definitions, which SMT solvers struggle to handle.
    \item Even after synthesizing a catamorphism \( \cata \), checking whether \( \cata \)
      preserves the validity of \(\bigwedge S\) %
      remains costly as it still involves recursively defined functions and ADTs.
\end{enumerate}
To address these challenges, we adopt an approach proposed by Reynolds et al.~\cite{reynolds2015counterexample}, a variant of counterexample-guided inductive synthesis (CEGIS), combined with a lightweight testing approach.

\subsubsection{Procedure \texttt{synthesis}}
\label{sec:synthesis}

\begin{algorithm}[t]
    \caption{Procedure \texttt{synthesis}}
    \label{alg:cegar-catamorphism}

    \begin{algorithmic}[1]
        \State \bf{Input: } $ \constraint $, $ \cata $,   $ \cache $, and $ \cataTMP $
        \State \bf{Output: } $ \cata' $ and $ \cache' $

        \State $\constraint' \gets$ a quantifier-free formula such that %
        $\constraint \equiv \forall \seq{x}.\: \constraint' $ %
        \State $\texttt{timeout} \gets \infty$
        \Loop
            \State $[\seq{x} \mapsto \seq{v}] \gets \text{check\_sat\_with\_TO}(\lnot \encoderalt{\constraint'}{\cata}, \texttt{timeout})$ \Comment{Testing}
            \State \textbf{if} $[\seq{x} \mapsto \seq{v}]$ is None or timeout \textbf{then} \Return $(\cata, \cache)$
            \State $\texttt{timeout} \gets  \texttt{defaultTimeout}$
            \State $\cache \gets \cache \land \constraint''$
            where $ \theta'' \equiv [\seq{v}/\seq{x}]\encoderalt{\constraint'}{\cataTMP} $
            \Comment{\( \cataTMP \) does not occur in \( \theta'' \) }
            \State $\model \gets \text{check\_sat}(\cache)$ \Comment{SAT modulo NIA}
            \State \textbf{if} \( \model \) is None \textbf{then return None}
            \State $\cata \gets \model(\cataTMP)$
        \EndLoop
    \end{algorithmic}
\end{algorithm}

Procedure~\ref{alg:cegar-catamorphism} shows the synthesis procedure.
This procedure takes a proof obligation \( \constraint \), a current catamorphism \( \cata \),  a constraint formula \( \cache \), and a template catamorphism \( \cataTMP \) as inputs.
Here, \( \constraint \) and \( \cata \) satisfy \( \not \models \encoderalt{\constraint}{\cata} \) but \( \models \constraint \).
The goal of the procedure is to find a new catamorphism \( \cata' \) that is \emph{likely} to satisfy \( \models \encoderalt{\constraint}{\cata} \).

As in ordinary CEGIS, the procedure consists of two phases: (a) verification (line 6) and (b) synthesis (line 10).
The former checks whether the current candidate of \( \cata \) satisfies \( \forall \tilde{x}.\: \constraint' \) by checking whether \( \lnot\encoderalt{ \constraint'}{\cata} \) is satisfiable.
If so, we obtain ground terms \( \tilde{v} \) such that \( [\tilde{v}/\tilde{x}]\encoderalt{\constraint'}{\cata} \) is invalid, and update \( \cache \) to (a formula equivalent to)
\( \cache \land [\tilde{v}/\tilde{x}]\encoderalt{\constraint'}{\cataTMP} \).
This enables us to synthesize a new catamorphism in a counterexample-guided manner on line 10.

A difference from the standard CEGIS is that we give up checking the satisfiability of
\(\lnot \encoderalt{\constraint'}{\cata}\)
upon a time-out on line 6. This is because an SMT prover is not good at proving the unsatisfiability of
\(\lnot \encoderalt{\constraint'}{\cata}\). If the satisfiability check times out, then
we optimistically assume that \(\encoderalt{\constraint'}{\cata}\) is valid, and returns the
current catamorphism \( \cata \) as a candidate solution.
In that case, it remains unknown whether \(\encoderalt{\constraint}{\cata} \) is indeed valid;
thus, the same counterexample  as the previous might be encountered in Procedure~\ref{alg:cegar-overview} again.

To ensure the progress even with this relaxation, we accumulate the necessary conditions \(\cache\) for the template parameters
of \(\cataTMP\) inside the synthesis loop (line 9) and return \(\cache\) to
Procedure~\ref{alg:cegar-overview}.
This allows Procedure~\ref{alg:cegar-catamorphism} to resume from the previous state whenever Procedure~\ref{alg:cegar-overview} encounters the same counterexample again.
For the progress, we must additionally require
that the given catamorphism is refined at least once during each synthesis call. %
To this end, we initialize the timeout to \( \infty \) (line 4) and later reset it to a default finite value (line 8).
This is justified by the fact that \( \encoderalt{\constraint}{\cata} \) for the given catamorphism \( \cata \) is already known to be invalid by Procedure~\ref{alg:cegar-overview}; therefore, an SMT solver should be able to find a model for \(\neg \encoderalt{\constraint'}{\cata} \) on line 6.

\begin{example}
    Let \( \cata \) be \( \banana{\combine_{\nil, 0}, \combine_{\cons, 0}} \) where \( \combine_{\nil, 0} = 0 \) and \( \combine_{\cons, 0}(x, l) = 0 \), \( \constraint \) be \( \forall l^{\dsort}.\: \constraint' \) where \( \constraint' \Def \lnot (\cons(0, \cons(0, l)) =_{\dsort} \nil) \),
    \( \cache\) be \( \true \),
    and \( \cataTMP \) be the linear template catamorphism in Example~\ref{ex:cata-linear-template-catamorphism} for \( \ilist \).
    We execute Procedure~\ref{alg:cegar-catamorphism} with these inputs: \( \constraint \), \( \cata \), \( \cache \), \( \cataTMP \).
        We first check whether \( \lnot\encoderalt{\constraint'}{\cata} \equiv \cata(\cons(0, \cons(0, l))) =_{\stypeint} \cata(\nil) \)
        is satisfiable, and find a model \( [l \mapsto \nil] \).
        As all arguments of sort \( \dsort \) in the catamorphism applications within \( [\nil / l]\encoderalt{\constraint'}{\cataTMP} \) are ground terms,
   \( [\nil / l]\encoderalt{\constraint'}{\cataTMP} \) can be simplified to:
    \(
        a \times (a \times d + c) + c \neq d
    \).
    We then update \( \cache \) to \(a \times (a \times d + c) + c \neq d\).
    Suppose an SMT solver yields the following model for \(\cache\):
    \(
        \model \equiv \set{
            a \mapsto 0, b \mapsto 0, c \mapsto 1, d \mapsto 0
        }.
    \)
    Based on this model, we have a new catamorphism \( \cata_1 \) defined by \( \banana{\combine_{\nil,1}, \combine_{\cons,1}} \) where
    \(
        \combine_{\nil} = 0
    \)
    and
    \(
        \combine_{cons}(x, l) = 1.
    \)
    Now the backend SMT solver either proves that \( \models \encoderalt{\constraint}{\cata_1} \) holds or times out. Therefore, we return \( \cata_1 \) and \( \cache \) as the result of Procedure~\ref{alg:cegar-catamorphism}.
\end{example}

%% file: disc.tex
\subsection{Discussions}
We discuss properties of the overall procedure of \catalia{} in this subsection.
We have already shown the soundness of the procedure in Section~\ref{sec:catalia-abstraction} (Theorem~\ref{th:cata-soundness}).
Other important questions are:
\begin{enumerate}
\item Relative completeness: Let \(\system\) be a system of CHCs,
  and suppose that \(\cataA(\system)\) is satisfiable for some \(\cataA\) (where \(\cata\) belongs to the class of
  catamorphisms expressed by a given set of templates).
  Assuming that the backend CHC solver over integers and SMT solver were sound and complete, does \catalia{} eventually prove that \(\system\) is satisfiable?
\item Refutational completeness:  Let \(\system\) be a system of CHCs,
  and suppose that \(\system\) is unsatisfiable.
  Assuming that the backend CHC solver over integers and SMT solver were sound and complete, does \catalia{} eventually prove that \(\system\) is unsatisfiable?
\end{enumerate}
We need to make some modifications and further assumptions to guarantee
that relative completeness and refutational completeness hold.

We can ensure relative completeness
by ensuring that 
the values of template parameters
  are chosen from a \emph{finite} set, and the set and the approximation degree \(N\)
  are gradually increased (as explained in Remark~\ref{rem:implementation-template})
  when there is no solution for the constraints in the current template,
  so that the whole class of catamorphisms being considered is eventually covered.
Note that given a template catamorphism, the same catamorphism is not encountered again
since the necessary conditions are accumulated in \(\Theta\) in Procedure~\ref{alg:cegar-catamorphism}.
Thus, if \(\cataA(\system)\) is satisfiable for some \(\cataA\) and such \(\cataA\) is an instance of
the current template catamorphism, it will eventually be found.

Refutational completeness holds if the underlying solver for CHC (over integers)
generates resolution proofs in a \emph{fair} manner, in the sense that,
given an infinite sequence of unsatisfiable CHCs \(\system_1,\system_2,\ldots\)
if every \(\system_i\) has a resolution proof of the same ``shape''
(except for constraint formulas), then that resolution proof is eventually produced.
That is because if \(\system\) is unsatisfiable, then its augmented version also has a resolution
proof, and all of its abstractions \(\cataAalt{\cata_1}(\system),\cataAalt{\cata_2}(\system),\ldots\) have the resolution proof of the same shape.
Thus, by the assumption of fairness, that resolution proof is eventually generated.

The assumption on the fairness above may be too strong in practice; in fact,
we do not think an existing CHC solver satisfies that property.
A more reasonable requirement would be to ensure that a CHC solver generates
a resolution proof of the smallest size.
Then, it suffices to ensure that a resolution proof of the same shape is never re-encountered,
e.g., by removing a prior bound on the approximation degree \(N\) and
removing the timeout on line 9 of Procedure~\ref{alg:cegar-catamorphism}.

%% file: eval.tex
\section{Implementation and Evaluation}
\label{sec:cata-eval}

We have implemented \catalia{}, a solver for the satisfiability checking problem of CHCs over ADT and LIA.
In this section, we describe the implementation details and evaluate \catalia{} on the benchmark set from CHC-COMP 2024~\cite{chc-comp-2024}.
As stated in the abstract, \catalia{} was used as a core part of the tool called
ChocoCatalia, which won the ADT-LIA category of CHC-COMP 2025 (which is the only category ChocoCatalia participated in).
ChocoCatalia relies on another independent, complementary
technique to be reported elsewhere, but \catalia{} alone would have won the competition, judging from
the evaluation result reported below for the CHC-COMP 2024 benchmark.

\subsection{Implementation}

The solver consists of the following four components:
\begin{enumerate}
    \item[(i)] \textbf{Preprocessing:} To handle testers and selectors that are not directly supported by our framework, we have implemented preprocessing steps described in Section 4 of \cite{KostyukovMF21}.
    \item[(ii)] \textbf{Abstraction:} This is an implementation of the procedure described in Section~\ref{sec:catalia-abstraction}. We utilized a portfolio of Spacer~\cite{spacer}, Eldarica~\cite{eldarica}, and \hoice{}~\cite{ChampionCKS20} as the backend solver for satisfiability checking problems of CHCs over integer arithmetic.
    \item[(iii)] \textbf{Counterexample Generation:} We extract a (hyper-)resolution proof generated by Z3/Spacer~\cite{MouraB08,spacer}, when a system of abstracted CHCs over integer arithmetic is unsatisfiable. We parse the result, and obtain a counterexample \( \constraint \) as described in Section~\ref{sec:cex-gen-overview}.
    \item[(iv)] \textbf{Refinement:} Our implementation of the synthesis procedure from Section~\ref{sec:cata-synthesis} employs multiple linear templates (see Remark~\ref{rem:implementation-template}).
    The testing component (line 6-8 in Procedure~\ref{alg:cegar-catamorphism}) is executed by running Z3 with a timeout of one second.
\end{enumerate}

While the formalization in the previous sections was for CHCs over a single ADT \( \dsort \),
the implementation can handle general CHCs over ADTs, including those with mutually recursive definitions. %

\subsection{Evaluation}

We evaluated \catalia{} on the ADT-LIA division of \emph{CHC-COMP 2024}~\cite{chc-comp-2024}, which consists of 300 benchmark instances of satisfiable and unsatisfiable CHCs over algebraic data types and linear integer arithmetic.
The benchmark set is publicly available~\cite{chc-comp-2024-benchmark}.
For comparison, we selected three state-of-the-art CHC solvers that support the theory of ADTs and LIA: Spacer~\cite{spacer}, RInGen~\cite{KostyukovMF21}, and Eldarica~\cite{eldarica}.
All the experiments were conducted on a machine with Intel Xeon Gold 6242 CPU and 64GB of RAM. We set the timeout to 300 seconds.
Further details of the evaluation are publicly available~\cite{SAS25Artifact}.

\begin{table}[t]
    \centering
\caption{Number of Solved Instances}
\label{tab:catalia-result}
    \begin{tabular}{lrrrrr}
      \toprule
      Instance & \catalia{} & \ringen{} & \spacer{} & \eldarica{} \\
      \midrule

      \# SAT & 67 & 54 & 48 & 50\\
      \# UNSAT & 80 & 46 & 86 & 87\\
      \# ALL & 147 & 100 & 134 & 137\\
      \midrule
      \# UNIQUE (SAT) & 18 & 14 & 3 & 1\\
      \# UNIQUE (UNSAT) & 2 & 3 & 4 & 1\\

      \bottomrule
      \end{tabular}
\end{table}

\begin{figure}[t]
    \centering
    \begin{subfigure}[t]{0.49\textwidth}
      \includegraphics[width=\textwidth]{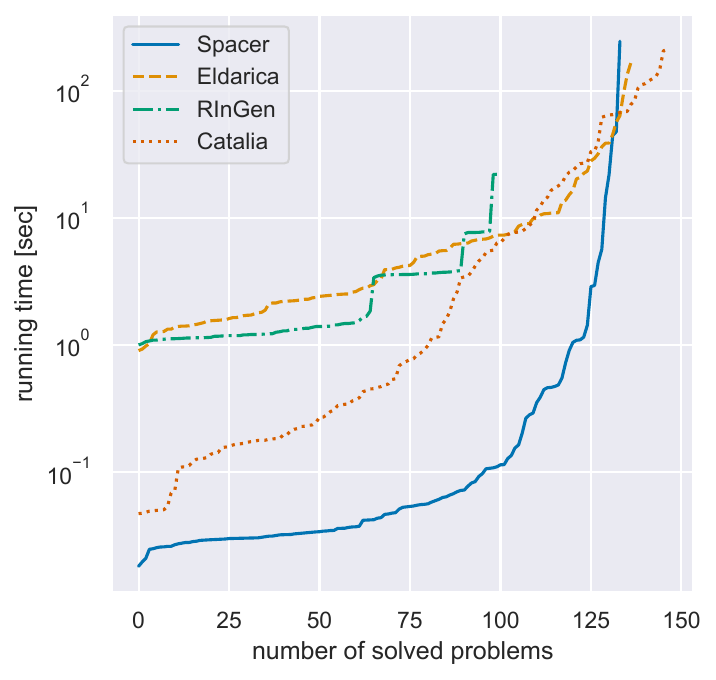}
      \caption{Results for All The Instances}
      \label{fig:catalia-result-all}
    \end{subfigure}\hfill
    \begin{subfigure}[t]{0.49\linewidth}
      \includegraphics[width=\textwidth]{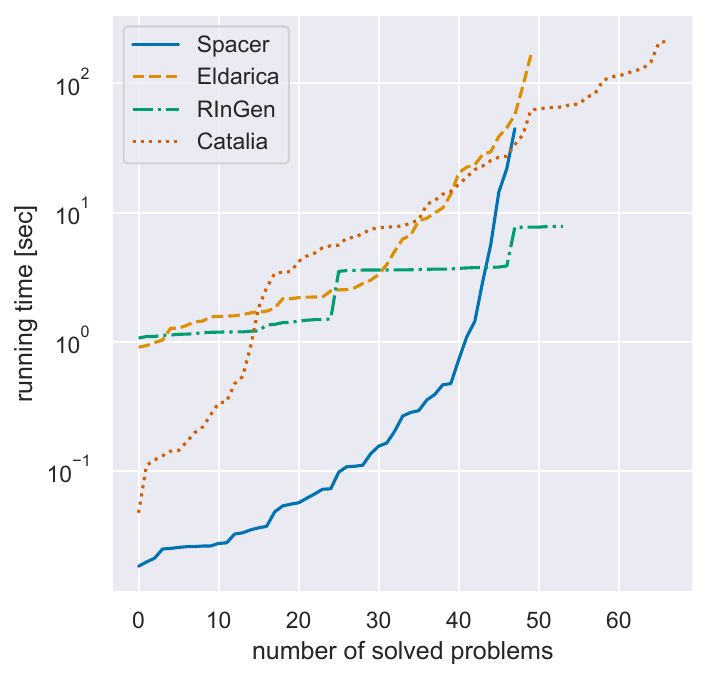}
      \caption{Results for SAT Instances}
      \label{fig:catalia-result-sat}
    \end{subfigure}
    \caption{Cactus plots. The horizontal axis shows the number of solved instances, and the vertical axis shows the time required to solve them.}
    \label{fig:catalia-result}
\end{figure}

The results, summarized in Table~\ref{tab:catalia-result} and Figure~\ref{fig:catalia-result}, show that \catalia{} performs
{particularly well on satisfiable instances}.
It solved 67 satisfiable instances, the most among all solvers.
Additionally, it uniquely solved 18 satisfiable instances that no other solver could solve successfully within 300 seconds.
Our approach primarily targets satisfiable CHCs, and the results confirm its effectiveness in this category.
Catamorphisms that \catalia{} successfully found include the list length, the sum of an integer list, the evenness of
the list length, and their combinations.
In terms of uniquely satisfiable instances, \catalia{} and RInGen are complementary because they handle different classes of invariants.
For example, \catalia{} can handle the list length, whereas RInGen cannot. Conversely, RInGen can handle invariants involving the last element of a list, which \catalia{} cannot.
The latter is due to \catalia{}'s restriction of template catamorphisms to linear ones (cf. Remark~\ref{rem:implementation-template}).
We expect that,  by extending templates with conditional expressions, \catalia{} will subsume RInGen's capability; we leave
this extension for future work.
For unsatisfiable instances, \catalia{} performed slightly worse than Spacer and Eldarica.
We also leave this issue for future work;
random testing techniques~\cite{KatsuraKSS24} may help address it.

Figure~\ref{fig:catalia-result} provides a detailed efficiency comparison across solvers.
Spacer demonstrates the fastest solving times overall, efficiently handling numerous instances.
However, its advantage lies primarily in speed rather than the number of solved satisfiable instances; in fact,
it uniquely solves only a few problems.
In contrast, \catalia{}, while slower in terms of solving time than Spacer, successfully solves more satisfiable instances, aligning with its design goal of handling more complex invariants.

%% file: related.tex
\section{Related Work}
\label{sec:rel}

We discuss related work on SMT solvers and CHC solvers that support ADTs.

\subsection{SMT Solvers}

The theory of ADT has been incorporated into the SMT-LIB Standard~\cite{BarFT-RR-17}, the de facto standard language specification for SMT solvers,
and leading SMT solvers such as Z3~\cite{MouraB08}, CVC5~\cite{BarbosaBBKLMMMN22} and Princess~\cite{princess08} already support this theory.
Since Oppen's work~\cite{Oppen80}, various decision procedures~\cite{ZhangSM06,SuterDK10,SuterKK11,Leino12,ReynoldsK15,PhamGW16,yang2019lemma} have been proposed to handle ADTs.

Among these, several approaches~\cite{SuterDK10,SuterKK11,PhamGW16,ReynoldsK15} address satisfiability modulo ADT and recursively defined functions (RDFs).
The approaches by Reynolds and Kuncak~\cite{ReynoldsK15} and Yang et al.~\cite{yang2019lemma} tackled automated inductive reasoning on ADTs and RDFs by an efficient enumeration of lemmas.
Suter et al.~\cite{SuterDK10,SuterKK11} introduced an abstraction method based on catamorphisms, which is similar to our approach.
They also proposed decision procedures based on the abstraction, which were later refined by Pham et al.~\cite{PhamGW16}.
Our abstraction, however, differs in that it is tailored for CHCs (e.g., by introducing \( \cata \)-admissibility predicates).
In particular, our approach is capable of automatically synthesizing catamorphisms, leveraging the result of CHC solving.

\subsection{CHC Solvers}

Various approaches have been proposed to solve CHCs over ADTs in order to capture more complex properties of ADTs~\cite{eldarica,AngelisFPP18a,AngelisFPP25,KostyukovMF21,KSG22,ZavaliaCF23,ChampionCKS20,KobayashiW23}.
Eldarica~\cite{eldarica} utilizes size constraints that represent the size of a given term to capture properties such as the list length.
Size functions, which are also utilized in the decision procedures by Zhang et al.~\cite{ZhangSM06}, can be seen as a special case of the catamorphisms introduced in this paper.
De Angelis et al~\cite{AngelisFPP18a} proposed fold/unfold transformation with techniques such as difference predicates~\cite{AngelisFPP20} and catamorphic abstractions~\cite{AngelisFPP25} to efficiently transform CHCs over ADTs to those without ADTs.
However, their solver is not capable of yielding a model even when it successfully proves the satisfiability, and it requires users to manually supply catamorphisms.
Kostyukov et al.~\cite{KostyukovMF21} reduced satisfiability checking of CHCs to finite model finding of first-order logic, by approximating constructors with uninterpreted functions.
A notable limitation of this approach is its inability to combine the theory of ADT with other theories such as linear integer arithmetic (LIA) and arrays.
Krishnan et al.~\cite{KSG22} have proposed a Spacer-like procedure for CHC over ADT and RDFs that preserves the refutational completeness of the original Spacer algorithm.
While their approach requires users to provide catamorphisms, ours automatically synthesizes them.
Some approaches~\cite{ZavaliaCF23,ChampionCKS20} transform predicates in CHCs to RDFs,
thereby reducing the problem to checking the satisfiability of formulas over ADTs and RDFs.
However, as discussed above, solving such a formula in SMT solvers can be challenging.
Furthermore, syntactically transforming predicates to functions is difficult especially when the CHCs are generated from compiler intermediate representations (e.g., LLVM), where the functional structure is often lost.
Kobayashi and Wu~\cite{KobayashiW23} employed a machine learning technique to synthesize inductive invariants over lists.
They train a recurrent neural network using an ICE learning framework and extract a fold (catamorphism) function using the technique proposed by Kobayashi et al.~\cite{KobayashiSSU21}.
Although their approach can, in theory, synthesize general recursive functions, its scalability remains a significant challenge.

%% file: conc.tex
\section{Conclusion}

\label{sec:conc}

We have proposed a method to solve the satisfiability checking problem of constrained Horn clauses over algebraic data types and integer arithmetic.
To find models defined inductively on the structure of algebraic data types, we employed
 catamorphisms to express inductive properties, and formalized a framework for automatically discovering appropriate catamorphisms on demand.
We also implemented a CHC solver \catalia{} based on the proposed method, and evaluated \catalia{} against the benchmark sets taken from CHC-COMP 2024 ADT-LIA division.
According to the evaluation results, \catalia{} outperformed the previous methods in solving SAT instances, indicating that \catalia{} is superior at discovering invariants that the previous solvers failed to find.

In future work, we plan to introduce more expressive catamorphism templates than linear ones.
As this may incur a cost in efficiency,
an important direction is to develop strategies for selectively applying different templates in the refinement phase of \catalia{}.

%% file: soundness2.tex
\section{Proof of Soundness}
\label{appx:soundness}
\newcommand*{\defeq}{\stackrel{\mathrm {def}}{=}}
\newcommand*{\cmap}{\bm{\gamma}}
\newcommand*{\absPred}{\AVAR{\PREDSYM}}
\newcommand*{\valuation}{\rho}
\newcommand*{\valuationTwo}{\zeta}
\renewcommand*{\term}{\dterm} %
\newcommand*{\cmodel}{\cmap(\model)}
\newcommand*{\valRel}[1][\cata, \CENV]{\approx_{#1}}

This section proves the soundness of the abstraction (Theorem~\ref{th:cata-soundness}).

\paragraph{Setting and Notation:}
Throughout this section we fix an approximation degree \( \NAPPROX \) and a catamorphism \( \cata \colon \herbrand_\dsort \to \Int^\NAPPROX \).
We also fix \( \system = \{ \clause_1, \ldots \clause_n \} \), a system of CHC over ADTs, and assume that \( \model \) is a model of \( \system \).
We suppose that for each predicate \( \PREDSYM \colon \overbrace{\dsort \times \cdots \times \dsort}^{k} \times \overbrace{\stypeint \times \to \times \stypeint}^l \to \Prop \) in the signature of ADT, there exists a predicate \( \absPred \) that takes \( \NAPPROX \times k + l \) integers in the signature of arithmetic.

Given a structure \( \mathcal{S} \), we write \( \sem{P}_{\mathcal S}\) and \( \sem{f}_{\mathcal S} \) for the interpretation of the predicate \( P \) and the function symbol \( f \) in \( \mathcal S \).
The denotation of a term and formula with respect to \( \mathcal{S} \) and a valuation \( \valuation \) are written as \( \sem{\dterm}_{\mathcal S, \valuation} \) and \( \sem{\varphi}_{\mathcal S, \valuation} (\in \{ \top, \bot \}) \), respectively; we may omit the subscripts if they are clear from the context.
As usual, we write  \( \mathcal{S}, \valuation \models \sem{\varphi} \) if \( \sem{\varphi}_{\mathcal S, \valuation} = \top \), and \( \mathcal{S} \models \varphi \) if \( \mathcal{S}, \valuation \models \varphi \) for all \( \valuation \).
\qed
\par
\medskip
The proof of soundness is given by constructing a model of \( \system \) from \( \model \), the model of \( \cataA(\system)\).
For a function (which may be thought of as an interpretation of an ``abstracted predicate'')
\[
f \colon
    \overbrace{\Int^\NAPPROX \times \dots \times \Int^\NAPPROX}^{k} \times \overbrace{\Int \times \dots \times \Int}^{l} \to \{\top, \bot\},
\]
we define its \emph{concretization} with respect to \( \cata \) by
\begin{gather*}
  \cmap_\cata(f) \colon \overbrace{\herbrand_\dsort \times \dots \times \herbrand_\dsort}^k \times \overbrace{\Int \times \dots \times \Int}^l \to \{ \top, \bot \} \\
  \cmap_\cata(f)(\dterm_1, \ldots, \dterm_k, n_1, \ldots, n_l) \defeq f(\cata(\dterm_1), \ldots, \cata(\dterm_k), n_1, \ldots, n_l)
\end{gather*}
We often omit the subscript \( \cata \) since it is fixed throughout this section.
Our goal is to show that concretizing the interpretation of \( \absPred \) gives a model of \( \system \).

\begin{definition}[Concretization of a Model]
    We write \( \cmodel \) for the structure over the signature for ADTs in which
  \begin{itemize}
    \item the universes for \( \dsort \) and \( \stypeint \) are \( \herbrand_\dsort \) and \( \Int \), respectively,
    \item the interpretation of \( \PREDSYM \) is given as \( \cmap(\sem{\absPred}_\model)\),
    \item the interpretation of \( =_\dsort \) is given as  the diagonal relation over \( \herbrand_\dsort \), and
    \item interpretation of function symbols is the same as that of \( \model \).
   \end{itemize}
 \end{definition}

Now we relate the interpretation of terms and formulas interpreted in \( \model \) and \( \cmap(\model )\).
To this end, we first prepare a relation over the valuations on \( \model \) and \( \cmodel \).
Let \( \CENV \) be a variable abstraction environment.
Given valuations \( \valuation \) and \( \valuationTwo \) over the (many-sorted) universes \( \langle \herbrand_\delta, \Int \rangle\) and \( \Int \), respectively, we write \( \valuation \valRel \valuationTwo \) if \( \cata(\valuation(x)) = \valuationTwo(\CENV(x)) \) for all \( x \in \dom(\CENV) \) and \( \valuation(x) = \valuationTwo(x) \) otherwise; here \( \valuationTwo(\CENV(x)) \) means \( (\valuationTwo(x_1), \ldots, \valuationTwo(x_n) ) \) for \( \CENV(x) = (x_1, \ldots, x_n) \).
\begin{lemma}
  \label{lem:soundness-term}
  Let \( \CENV \) be a variable abstraction environment and let \( \dterm \) be a term of sort \( \dsort \) such that \( \fv(\dterm) \subseteq \dom(\CENV) \).
  If \( \valuation \valRel \valuationTwo \), then we have \(  \cata(\sem{\dterm}_{\valuation}) = \sem{\cataAT(\dterm)}_\valuationTwo  \).
\end{lemma}
\begin{proof}
  By induction on the structure of \( \dterm \).
  The case \( \dterm = x \) with \( x \in \dom(\CENV) \) holds because
  \begin{align*}
    \cata(\sem{x}_{\valuation}) = \cata(\valuation(x)) = \valuationTwo(\CENV(x)) = \sem{\cataAT(x)}_\valuationTwo
  \end{align*}
  by the definition of \( \valuation \valRel \valuationTwo \).

  We next show the case \( \dterm = \CONS_i(\dterm_1, \dots, \dterm_{m_i}, \arith_1, \dots, \arith_{n_i}) \).
  We have
      \begin{align*}
        &\sem{\cataAT(\CONS_i(\dterm_1, \dots, \dterm_{m_i}, \arith_1, \dots, \arith_{n_i}))}_\valuationTwo \\
        &= \combine_i(\sem{\cataAT(\dterm_1)}_\valuationTwo, \dots, \sem{\cataAT(\dterm_{m_i})}_\valuationTwo, \sem{\arith_1}_\valuationTwo, \dots, \sem{\arith_{m_i}}_\valuationTwo) \\
        &= \combine_i(\cata(\sem{\dterm_1}_\valuation), \dots, \cata(\sem{\dterm_{m_i}}_\valuation), \sem{\arith_1}_\valuation, \dots, \sem{\arith_{m_i}}_\valuation) \tag{by I.H.} \\
        &= \cata(\CONS_i(\sem{\dterm_1}_\valuation, \dots, \sem{\dterm_{m_i}}_\valuation, \sem{\arith_1}_\valuation, \dots, \sem{\arith_{m_i}}_\valuation))\tag{by def.\ of \( \cata\)} \\
        &= \cata(\sem{\dterm}_\valuation) \tag{by def.~of \( \sem{\dterm} \)}
      \end{align*}
   as desired.

  The other cases are trivial.
  \qed
\end{proof}
\begin{lemma}
  \label{lem:soundness-pred}
  Suppose that \( \valuation \valRel \valuationTwo \).
  Then we have:
  \begin{enumerate}
    \item \( \cmodel, \valuation \models \PREDSYM_i(\seq{\dterm})\) if and only if \( \model, \valuationTwo \models \absPred_i(\seq{\cataAT(\dterm)}) \)
    \item if \( \cmodel, \valuation  \models  t_1 =_{\dsort} t_2 \), then \(  \model, \valuationTwo \models \cataAT(\dterm_1) =_{\stypeint^{\NAPPROX}} \cataAT(\dterm_2) \)
  \end{enumerate}
\end{lemma}
\begin{proof}
  We only prove 1.\ since 2.\ is trivial.
  Suppose that \( \seq \dterm = \dterm_1, \ldots, \dterm_k, \allowbreak \arith_1, \ldots, \arith_l \), where each \( \dterm_i \) has sort \( \dsort \) and \( \arith_i \) is of sort \( \stypeint \).
  Then we have
  \begin{align*}
    &\sem{\PREDSYM_i(\seq{\dterm})}_{\cmodel, \valuation} \\
    &= \cmap(\sem{\PREDSYM_i}_{\model})(\sem{\dterm_1}_{\valuation}, \ldots, \sem{\dterm_k}_{\valuation},   \sem{\arith_1}_\valuation, \ldots, \sem{\arith_k}) \\
    &= \sem{\PREDSYM_i}_{\model}(\cata(\sem{\dterm_1}_{\valuation}), \ldots, \cata(\sem{\dterm_k}_{\valuation}),   \sem{\arith_1}_\valuation, \ldots, \sem{\arith_k}) \tag{by the def.\ of \( \cmap \)} \\
    &=\sem{\PREDSYM_i}_{\model}\sem{\cataAT(\dterm_1)}_{\valuationTwo}), \ldots, \sem{\cataAT(\dterm_k)}_{\valuationTwo}),   \sem{\cataAT(\arith_1)}_\valuationTwo, \ldots, \sem{\cataAT(\arith_k)}_\valuationTwo) \tag{by Lemma\ref{lem:soundness-term} and \( \cataAT(\arith_i) = \arith_i\)} \\
    &= \sem{\absPred_i(\seq{\cataAT(\dterm)})}_{\model, \valuationTwo}
  \end{align*}
  \qed
\end{proof}

\begin{lemma}
  \label{lem:soundness-constraint}
  Let \( \CENV \) be a variable abstraction environment and assume that \( \valuation \valRel \valuationTwo \).
  Then \( \cmodel, \valuationTwo \models \constraint \) implies \( \model, \valuation \models \cataAFB(\constraint) \).
\end{lemma}
\begin{proof}
  By straightforward induction on the structure of \( \constraint \). The case for \( \dterm_1 =_{\dsort} \dterm_2 \) follows from Lemma~\ref{lem:soundness-pred}.
  \qed
\end{proof}

We show that \( \absPred_\dsort(x_1, \ldots, x_\NAPPROX) \), the abstraction of \( \PREDSYM_\dsort(x) \) represents the image of \( \cata \).
\begin{lemma}
  \label{lem:admissibility-appx}
  \newcommand*{\minModel}{\model_{\mathrm{min}}}
  \newcommand*{\cataSort}{\PREDSYM_\dsort}
  \newcommand*{\cataSorts}{\seq{\PREDSYM}_\dsort}
  Let \( \cataSorts \) be the system of CHCs defining \( \cataSort \).
  Suppose that  \( \minModel \) be the minimum model of \( \cataA(\cataSorts) \)
  Then, \( \sem{\absPred_\dsort}_{\minModel}(n_1, \dots, n_\NAPPROX) = \top \)  if and only if there exists \( t \in \universe{\herbrand}_{\dsort} \) such that \( (n_1, \dots, n_\NAPPROX) = \cata(t) \).
\end{lemma}
\begin{proof}
  \newcommand*{\minModel}{\model_{\mathrm{min}}}
  \newcommand*{\cataSort}{\PREDSYM_\dsort}
  \newcommand*{\cataSorts}{\seq{\PREDSYM}_\dsort}
  Observe that for each constructor \( \CONS_i \), we have the following clause in \( \cataA(\cataSorts) \).
  \begin{align*}
    \forall \seq{y_1} \ldots \seq{y_k} \:\seq{x}. \absPred_\dsort(\combine_i(\seq{y_1}, \ldots, \seq{y_k}, \seq{x})) \Longleftarrow \absPred_\dsort(\seq{y_1}) \land \cdots
    \land \absPred_\dsort(\seq{y_k})
  \end{align*}
  where \( \seq{y_i} \) is a sequence of integer variables \( y_{i1}, \ldots, y_{i\NAPPROX} \).

  One can easily check that the structure \( \minModel \) defined by
  \begin{align*}
    \sem{\AVAR{\cataSort}}_{\minModel}(n_1, \ldots, n_\NAPPROX) =
    \begin{cases}
      \top \quad \text{if \( (n_1, \ldots, n_\NAPPROX) = \cata(\dterm) \) for some \( \dterm \)}\\
      \bot \quad \text{otherwise}
    \end{cases}
  \end{align*}
  is a model by checking the denotation of the each clause of \( \cataA(\cataSorts) \).

  To check that this is indeed a minimum model, it suffices to show that for every model \( \model \) of \( \cataA(\cataSorts) \) and every term \( \dterm \in \herbrand_\dsort\), we have \( \sem{\AVAR{\cataSort}}_{\model}(\cata(\dterm)) = \top \).
  This is proved by induction on the size of \( \dterm \) with a case analysis on the shape of \( \dterm \).
  Suppose that \( \dterm = \CONS_i(\dterm_1, \ldots, \dterm_k, \arith_1, \ldots, \arith_l)\).
  Since \( \model \) is a model of \( \cataA(\cataSorts) \), we have
  \begin{align*}
    \model, \valuation \models \AVAR{\cataSort}( \combine_i(\seq{y_1}, \dots, \seq{y_k}, \seq{x})) \Longleftarrow \AVAR{\cataSort}(\seq{y_1}) \land \dots \land \AVAR{\cataSort}(\seq{y_k})
  \end{align*}
  where \( \valuation = [\seq{y_1} \mapsto \cata(\dterm_1), \ldots \seq{y_k} \mapsto \cata(\dterm_k), \seq{x} \mapsto \seq{\sem{\arith}_{\model}}]\).
  We have \( \sem{\AVAR{\cataSort}}_{\model}(\cata(\dterm_i))  \allowbreak = \top \) for \( i \in \{ 1, \ldots, k \}\) by the induction hypothesis, which implies \( \model, \valuation \models \AVAR{\cataSort}( \combine_i(\seq{y_1}, \dots, \seq{y_k}, \seq{x})) \).
  From this and \( \cata(\dterm) = \combine_i(\cata(\dterm_1), \ldots, \cata(\dterm_k), \sem{\arith_1}, \ldots, \sem{\arith_l})\), we must have  \( \sem{\AVAR{\cataSort}}_{\model}(\cata(\dterm)) = \top \).
  \qed
\end{proof}

Finally, we show that \( \cmodel \) is indeed a model.
\begin{theorem}
 The structure \( \cmodel \) is a model of \( \system = \{ \clause_1, \ldots, \clause_n \} \).
\end{theorem}
\begin{proof}
  It suffices to prove that \( \cmodel \models \clause_i \) for every \( i \).
  Suppose that \( \clause_i \) is of the form
  \begin{align*}
   \forall x_1^\dsort \ldots x_p^\dsort  y^\stypeint_1 \ldots y_q^\stypeint . \PREDSYM_{i_0}(\seq {\dterm_{i_0}}) \Longleftarrow \constraint \land \PREDSYM_{i_1}(\seq {\dterm_{i_1}}) \land \cdots \land \PREDSYM_{i_r}(\seq {\dterm_{i_r}}).
  \end{align*}
  Here we have assumed that the head of the clause is a predicate application, but the case where the head is \( \textit{false} \) can be proved in similar manner.
  The clause \( \cataA(\clause_i) \) is then of the form\footnote{In this proof universal quantifiers are explicitly written.}
  \begin{align*}
    \forall \seq{x_1} \ldots \seq{x_p}, y_1 \ldots y_q. \absPred_{i_0}(\seq {\cataAT(\dterm_{i_0})}) \Longleftarrow&\
    \begin{aligned}[t]
     &\cataAFB(\constraint) \land \varphi_{\mathrm{adm}} \\
     &\land \absPred_{i_1}(\seq {\cataAT(\dterm_{i_1})}) \land \cdots \land \absPred_{i_r}(\seq {\cataAT(\dterm_{i_r})})
    \end{aligned}\\
    &\text{where } \varphi_{\mathrm{adm}} \defeq \absPred_\dsort(\seq{x_1}) \land \cdots \land \absPred_\dsort(\seq{x_p}).
  \end{align*}
  Here \( \seq{x_i} = x_{i1}, \ldots, x_{i\NAPPROX} \) and \( \CENV = [x_1 \mapsto \seq{x_1}, \ldots, x_p \mapsto \seq{x_p}]\).

  Our goal is to show \( \cmodel, \valuation \models \PREDSYM_{i_0}(\seq {\dterm_{i_0}}) \Longleftarrow \constraint \land \PREDSYM_{i_1}(\seq {\dterm_{i_1}}) \land \cdots \land \PREDSYM_{i_r}(\seq {\dterm_{i_r}}) \) for all valuation \( \valuation \).
  Let \(  \dterm_1, \ldots, \dterm_p \in \herbrand_\dsort \), \( n_1, \ldots, n_q \) and
  \begin{align*}
    \valuation \defeq [x_1 \mapsto \term_1 , \ldots, x_p \mapsto \dterm_p, y_1 \mapsto n_1, \ldots, y_q \mapsto n_q].
  \end{align*}
  Let \( \valuationTwo \) be a valuation defined by
  \begin{align*}
    \valuationTwo \defeq [\CENV(x_1) \mapsto \cata(\dterm_1), \ldots, \CENV(x_p), \mapsto \cata(\dterm_p), y_1 \mapsto n_1, \ldots, y_q \mapsto n_q]
  \end{align*}
  Clearly, we have \( \valuation \valRel \valuationTwo\).
  By the assumption that \( \model \) is a model of \( \cataA(\system)\), we have \( \model \models \cataA(\clause_i) \).
  In particular, we have
  \begin{align*}
    \model, \valuationTwo \models \absPred_{i_0}(\seq {\cataAT(\dterm_{i_0})}) \Longleftarrow&\  \cataAFB(\constraint) \land \absPred_{i_1}(\seq {\cataAT(\dterm_{i_1})}) \land \cdots \land \absPred_{i_r}(\seq {\cataAT(\dterm_{i_r})}) \land \varphi_{\mathrm{adm}}
  \end{align*}
  Thanks to Lemma~\ref{lem:admissibility-appx}, we also know that \( \model, \valuationTwo \models \varphi_{\mathrm{adm}} \).
  From this, it follows that
  \begin{align*}
    \model, \valuationTwo \models \absPred_{i_0}(\seq {\cataAT(\dterm_{i_0})}) \Longleftarrow&\  \cataAFB(\constraint) \land \absPred_{i_1}(\seq {\cataAT(\dterm_{i_1})}) \land \cdots \land \absPred_{i_r}(\seq {\cataAT(\dterm_{i_r})})
  \end{align*}
  Applying Lemma~\ref{lem:soundness-pred} and~\ref{lem:soundness-constraint} to the above satisfaction relation, we get
  \begin{align*}
    \cmodel, \valuation \models \PREDSYM_{i_0}(\seq {\dterm_{i_0}}) \Longleftarrow \constraint \land \PREDSYM_{i_1}(\seq {\dterm_{i_1}}) \land \cdots \land \PREDSYM_{i_r}(\seq {\dterm_{i_r}})
  \end{align*}
  as desired.
  \qed
\end{proof}

%% file: app_exp.tex
\section{Additional Material for Evaluation}
\label{ap:additional-eval}

\begin{table}[t]
    \centering
    \caption{Tools used in the evaluation}
    \begin{tabular}{l|l|l}
        Tool & Version  \\
        \hline
        Z3 & 4.12.2 \\
        \spacer{} & Z3 version 4.12.2 \\
        \eldarica{} & v2.1 \\
        \hoice{} & 1.10.0 \\
        Vampire & 4.6.1 (Z3: 4.8.13.0) \\
    \end{tabular}
    \label{tab:tools}
\end{table}

We listed the versions of the tools used in our evaluation in Table~\ref{tab:tools}.
We used the following option for z3 to use Spacer:
\begin{itemize}
  \item \texttt{fp.xform.tail\_simplifier\_pve=false}
  \item \texttt{fp.validate=true}
  \item \texttt{fp.spacer.mbqi=false}
  \item \texttt{fp.spacer.use\_iuc=true}
  \item \texttt{fp.spacer.global=true}
  \item \texttt{fp.spacer.expand\_bnd=true}
  \item \texttt{fp.spacer.q3.use\_qgen=true}
  \item \texttt{fp.spacer.q3.instantiate=true}
  \item \texttt{fp.spacer.q3=true}
  \item \texttt{fp.spacer.ground\_pobs=false}
\end{itemize}
which can be found in the competition report for CHC COMP 2021~\cite{chc-comp-2021}.
We enabled \texttt{-portfolio} option for \eldarica{}, \texttt{-t --no-isolation --solver vampire} option for RInGen.
We utilized the default options for the other tools.